\documentclass[10pt]{amsart}
\usepackage[T1]{fontenc}
\usepackage{amsmath,amsfonts,amstext,amssymb,amsthm,
flexisym,verbatim}
\usepackage{braket,mleftright}
\usepackage{mathrsfs}
\usepackage{dsfont}
\usepackage{cite}
\usepackage[colorlinks=true,linkcolor=blue,urlcolor=blue,
citecolor=blue,anchorcolor=blue]{hyperref}
\newcommand{\R}{\mathbb{R}}%reals
\newcommand{\C}{\mathbb{C}}%complexes
\newcommand{\N}{\mathbb{N}}%naturals
\newcommand{\Z}{\mathbb{Z}}%integers
\newcommand{\mrm}[1]{\mathrm{#1}}%roman math
\newcommand{\ol}[1]{\overline{#1}}%long bar
\newcommand{\co}{\colon}%colon
\newcommand{\vrt}{\,\vert\,}%vertical line
\newcommand{\lto}{\rightarrow}%space map 
\newcommand{\lmap}{\mapsto}%function map 
\newcommand{\abs}[1]{\lvert#1\rvert}%absolute value
\newcommand{\norm}[1]{\lVert#1\rVert}%norm
\newcommand{\img}{\mrm{i}}%imaginary unit
\newcommand{\e}{\mrm{e}}%Euler exponent
\newcommand{\ie}{\textit{i.e.}\;}%"that is"
\newcommand{\eg}{\textit{e.g.}\;}%"for example"
\newcommand{\Hs}{\mathfrak{h}}%Hilbert space L^2(\R^3)
\newcommand{\Hl}{\mathfrak{H}}%Hilbert space L^2(\R^3)\ot \C^2
\renewcommand{\H}{\mathcal{H}}%Hilbert space
\newcommand{\He}{\mathbb{H}}%extension of H
\newcommand{\Uso}{U^{\mrm{so}}}%spin-orbit term
\newcommand{\UB}{U^{\mrm{B}}}%magnetic term
\newcommand{\cgc}[6]{\begin{bmatrix}
#1 & #2 & #3 \\ #4 & #5 & #6\end{bmatrix}}%Clebsch-Gordan coefficient
\newcommand{\K}{\mathfrak{K}}%Hilbert space (\Hs\ot\Hs)\ot(\C^3\op\C^1)
\newcommand{\bas}{\mathcal{S}}%the set of (S,s)
\newcommand{\id}{\mathds{1}}%identity map
\newcommand{\setm}{\smallsetminus}%set minus
\newcommand{\D}{\mathcal{D}}%the set 
\newcommand{\B}{\mathbf{B}}%the algebra of bounded operator
\newcommand{\Ds}{\mathscr{D}}%the space of Schwartz distributions
\newcommand{\F}{\mathscr{F}}%Fourier transform
\newcommand{\ot}{\otimes}%tensor product
\newcommand{\op}{\oplus}%direct sum
\newcommand{\pd}{\partial}%partial derivative
\newcommand{\dd}{\,\mrm{d}}%differential
\DeclareMathOperator{\dom}{dom}%domain
\DeclareMathOperator{\res}{res}%resolvent set
\DeclareMathOperator*{\LIM}{l.i.m.}%limit
\DeclareMathOperator*{\slim}{\mathrm{s-}\!\lim}%strong limit
\theoremstyle{plain}
\newtheorem{thm}{Theorem}[section]
\newtheorem{lem}[thm]{Lemma}
\newtheorem{prop}[thm]{Proposition}
\newtheorem{cor}[thm]{Corollary}

\theoremstyle{definition}
\newtheorem{rem}[thm]{Remark}
\allowdisplaybreaks
\numberwithin{equation}{section}
\begin{document}
\title[Computation of unitary group for the Rashba operator]{Computation of 
		unitary group for the Rashba spin-orbit 
		coupled operator, with application to point-interactions}
\author{Rytis Jur\v s\. enas}
\address{Vilnius University, Center for Physical Sciences and Technology,   
		Institute of Theoretical Physics and Astronomy,  
		Saul\.{e}tekio ave.~3, Vilnius 10222, Lithuania}
\email{Rytis.Jursenas@tfai.vu.lt}
\date{\today}
%%%%%%%%%%%%%%%%%%%%%%%%%%%%%%%%%%%%%%%%%%%%%%%%%%%%%%%%%%%%%%
%%%%%%%%%%%%%%%%%%%%%%%%%%%%%%%%%%%%%%%%%%%%%%%%%%%%%%%%%%%%%%
%%%%%%%%%%%%%%%%%%%%%%%%%%%%%%%%%%%%%%%%%%%%%%%%%%%%%%%%%%%%%%
%%%%%%%%%%%%%%%%%%%%%%%%%%%%%%%%%%%%%%%%%%%%%%%%%%%%%%%%%%%%%%
\begin{abstract}
     We compute an explicit formula for the one-parameter unitary 
     group of the single-particle Rashba spin-orbit coupled operator 
     in dimension three. As an application, we derive the formula for 
     the Green function for the two-particle operator, and then prove 
     that the spin-dependent point-interaction is of class $\H_{-4}$. 
     The latter is thus the example of a supersingular perturbation 
     for which no self-adjoint operator can be constructed.
\end{abstract}
%%%%%%%%%%%%%%%%%%%%%%%%%%%%%%%%%%%%%%%%%%%%%%%%%%%%%%%%%%%%%%
%%%%%%%%%%%%%%%%%%%%%%%%%%%%%%%%%%%%%%%%%%%%%%%%%%%%%%%%%%%%%%
%%%%%%%%%%%%%%%%%%%%%%%%%%%%%%%%%%%%%%%%%%%%%%%%%%%%%%%%%%%%%%
%%%%%%%%%%%%%%%%%%%%%%%%%%%%%%%%%%%%%%%%%%%%%%%%%%%%%%%%%%%%%%
%\pacs{71.70.Ej, 03.65.Db, 02.30.Tb}
\maketitle
%%%%%%%%%%%%%%%%%%%%%%%%%%%%%%%%%%%%%%%%%%%%%%%%%%%%%%%%%%%%%%
%%%%%%%%%%%%%%%%%%%%%%%%%%%%%%%%%%%%%%%%%%%%%%%%%%%%%%%%%%%%%%
%%%%%%%%%%%%%%%%%%%%%%%%%%%%%%%%%%%%%%%%%%%%%%%%%%%%%%%%%%%%%%
%%%%%%%%%%%%%%%%%%%%%%%%%%%%%%%%%%%%%%%%%%%%%%%%%%%%%%%%%%%%%%
\section{Introduction}
%%%%%%%%%%%%%%%%%%%%%%%%%%%%%%%%%%%%%%%%%%%%%%%%%%%%%%%%%%%%%%
%%%%%%%%%%%%%%%%%%%%%%%%%%%%%%%%%%%%%%%%%%%%%%%%%%%%%%%%%%%%%%
%%%%%%%%%%%%%%%%%%%%%%%%%%%%%%%%%%%%%%%%%%%%%%%%%%%%%%%%%%%%%%
%%%%%%%%%%%%%%%%%%%%%%%%%%%%%%%%%%%%%%%%%%%%%%%%%%%%%%%%%%%%%%
The fundamental object for describing quantum dynamics of a 
system governed by a Hamiltonian $h$ is the associated 
one-parameter unitary group $\e^{\img t h}$ 
($t\in\R$). For $h$ lower semibounded, the unitary group is 
closely related to the semigroup $\e^{-th}$ ($t>0$) in that 
$\e^{-\img t h}$ is the strong limit of 
$\e^{-\img(t-\img\epsilon)h}$ as $\epsilon\searrow0$. For example, 
the integral kernel (the free propagator) of the Schr\"{o}dinger 
semigroup $\e^{t\Delta\vert_{H^2(\R^3)}}$ is well-known:
\begin{equation}
K_t^0(x)=\frac{\e^{-\frac{\abs{x}^2}{4t}}}{(4\pi t)^{3/2}}
\quad \text{a.e. $x\in\R^3$}
\label{eq:1}
\end{equation}
and
\begin{equation}
\e^{\img t\Delta\vert_{H^2(\R^3)}}u
=\LIM\int K_{\img t}^0(\cdot-y)u(y)\dd y
\label{eq:2}
\end{equation}
for $u\in L^2(\R^3)$. Here and elsewhere 
$\LIM\int\equiv\LIM_{R\nearrow\infty}\int_{\abs{y}\leq R}$ 
implies the $L^2$-norm convergent integral; the integral is just the 
convolution $K^0_{\img t}*u$ when $u\in L^2\cap L^1$. 
For a comprehensive exposition of Schr\"{o}dinger (semi)groups 
the reader may refer to \cite{Simon-2,Yosida,Simon-1,Simon-3}.

In the first half of the present paper we obtain the formulas 
(theorem~\ref{thm:lE}, corollary~\ref{cor:lE}) analogous to 
\eqref{eq:1} and \eqref{eq:2} when $h$ is the Rashba spin-orbit coupled 
operator considered in the presence of the out-of-plane magnetic field. 
To the best of our knowledge, no such formula has been derived previously. 
More specifically, we consider the operator $h$ in 
$\Hl=\Hs\ot\C^2$, with $\Hs=L^2(\R^3)$, given by the operator sum 
\begin{subequations}\label{eq:h}
\begin{equation}
h=h^0+U\,, \quad h^0=-\Delta\vert_{H^2(\R^3)}\ot\id_{\C^2}
\end{equation}
where the potential $U$ (atom-light coupling) is the operator sum
\begin{equation}
U=\alpha\Uso+\beta\UB\,, \quad \alpha,\,\beta\geq0
\end{equation}
of the Rashba spin-orbit coupling term
\begin{equation}
\Uso=\ol{D^-\ot S^+-D^+\ot S^-}\,, \quad 
D^\pm=(\nabla_1\pm \img\nabla_2)\vert_{H^1(\R^3)}
\end{equation}
(the overbar denotes the closure) and the Raman-coupling term 
\begin{equation}
\UB=2\id_\Hs\ot S^3\,.
\end{equation}
\end{subequations}
The spin operators 
\begin{equation}
S^i=\frac{1}{2}\sigma^i \quad\text{for}\quad i\in\{1,2,3\}\,; \quad
S^\pm=S^1\pm \img S^2
\label{eq:spins}
\end{equation}
where $(\sigma^i)$ are standard Pauli matrices.

By using a classic Nelson theorem for analytic vectors we show that 
$D^-\ot S^+-D^+\ot S^-$ is essentially self-adjoint and that the 
linear span $\D$ of bounded elementary tensors $u\ot\varphi$, with 
$u\in\D_b(D^+D^-)$ and $\varphi\in\C^2$, is the core for its closure 
$\Uso$. As a result, the operator $h$ defined by \eqref{eq:h}, 
\eqref{eq:spins} is self-adjoint. Of course, to show the 
self-adjointness of $h$ it would suffice to notice that $h^0$ 
is self-adjoint and that $U$ is relatively $h^0$-bounded. However, 
it is $\D$ itself that is more important to our investigation 
besides the self-adjointness.

It is known that $h\geq-\Sigma$ is lower semibounded with 
$\Sigma=\beta$ if $2\beta>\alpha^2$ and 
$\Sigma=(\beta/\alpha)^2+(\alpha/2)^2$ otherwise. 
Variant forms of $h$ were studied by many authors. For example, 
the special case $\alpha=0$ in various spatial dimensions can be 
found in \cite{Carlone11,Cacciapuoti09,Cacciapuoti07,Bruning07,Exner07}. 
A general case in three spatial dimensions is studied in 
\cite{Jursenas16,Jursenas14}. 

In the second half of the paper we use the computed unitary group 
for deriving the Green function for the two-particle operator
\begin{equation}
H=h\ot \id_\Hl+\id_\Hl\ot h\,.
\label{eq:H0}
\end{equation}
Since $h\geq-\Sigma$ is self-adjoint in $\Hl$, 
$H\geq-2\Sigma$ is self-adjoint in $\Hl\ot\Hl$. We use the integral 
representation of the resolvent $R(z)=(H-z\id_{\Hl\ot\Hl})^{-1}$,
\begin{equation}
R(z)=\pm\img\int_0^\infty\e^{\pm\img tz}\e^{\mp\img tH}\dd t\,, \quad
\e^{\img tH}=\e^{\img th}\ot\e^{\img th}
\label{eq:resT-2}
\end{equation}
where the upper (lower) sign is taken when $\Im z>0$ ($\Im z<0$). 
For $\alpha$ small, we give explicit formulas for the parts of 
Green function in propositions~\ref{prop:RcAlpha0}, 
\ref{prop:RcDiagAlphaSmall}, and \ref{prop:RcNonDiagAlphaSmall}.

Our main motive for considering \eqref{eq:H0} is an attempt to 
understand eventually in a rigorous way the formation of cold 
molecules \cite{Fu}, provided that the interaction is zero-range. 
The main difficulty is that the two-particle operator $H$ does not admit 
the separation of variables in the center-of-mass coordinate system 
$Q=(x,X)$ unless $\alpha=0$; here $x$ is the relative coordinate and 
$X$ is the center-of-mass coordinate. The situation is clearly seen from 
the operator which is unitarily equivalent to $H$:
\begin{equation}
A\ot \id_\Hl+\id_\Hl\ot B+\alpha \ol{D}
\label{eq:Op-ABD}
\end{equation}
where the self-adjoint operators $A=A(x)$ and $B=B(X)$ are defined by
\begin{subequations}\label{eq:ABD}
\begin{equation}
A=2h^0+U\,, \quad B=\frac{1}{2}(h^0+U+\beta\UB)
\end{equation}
and the essentially self-adjoint operator $D=D(x,X)$ is defined by 
\begin{align}
D=&(D^+\ot \id_{\C^2})\ot(\id_\Hs\ot S^-)
-(D^-\ot \id_{\C^2})\ot(\id_\Hs\ot S^+) 
\nonumber \\
&+\frac{1}{2}[(\id_\Hs\ot S^+)\ot(D^-\ot \id_{\C^2})
-(\id_\Hs\ot S^-)\ot(D^+\ot \id_{\C^2})]\,.
\end{align}
\end{subequations}
When $\alpha=0$, \eqref{eq:Op-ABD} simplifies so that one 
can apply the theory of singular perturbations developed 
in \cite{Albeverio00} (and in particular in theorem~5.2.1 therein), 
since it was shown in \cite{Jursenas16} that the Dirac delta is of 
class $\H_{-2}(h)\setm\H_{-1}(h)$ for $h$, and hence for 
$A=2h^0+\beta\UB$; as usual, $(\H_n)_{n\in\Z}$ is the scale of 
Hilbert spaces associated with a self-adjoint operator. We assume 
that the two particles at positions $x_1$ and $x_2$ are interacting 
via the zero-range potential which is modeled by the Dirac distribution 
concentrated at $x=x_1-x_2=0$.

When $\alpha>0$, one cannot associate the $x$-dependent 
interaction potential to $A(x)$ alone because of $D(x,X)$; 
that is, the two-particle case no longer reduces to the 
single-particle one. Instead, one studies the restriction of 
\eqref{eq:Op-ABD} to $(H^2(\R^3\setm\{0\})\ot\C^2)\ot\dom h$. 
Equivalently, one considers the singular perturbation concentrated 
at $(0,X)$ and associated to the total operator \eqref{eq:Op-ABD}. 
Further results in this direction will be provided elsewhere. 
Here, we aim at considering the perturbation itself, and we show 
by using \eqref{eq:resT-2} that 
%for $\alpha\geq0$ sufficiently small, 
the perturbation is of class $\H_{-4}(H)\setm\H_{-3}(H)$; 
see theorem~\ref{thm:super}. 
A general theory of supersingular rank one perturbations is developed 
in \cite{Kurasov2,Kurasov} (see also the list of references therein), 
where it is shown that the restricted operator does not have 
self-adjoint extensions, but the so-called regular ones. The results 
for finite rank perturbations, which is our case, are generalized 
naturally.

One can also obtain the Green function for \eqref{eq:H0} by using 
the single-particle Green function in \cite{Jursenas14}. 
Mimicking the proof of theorem~5 in \cite{Simon-4}, 
for $z\in\C$ and $\Re z<-2\Sigma$, the resolvent $R(z)$ of $H$ is 
the norm convergent integral 
$\int_\R r(z/2-\img t)\ot r(z/2+\img t)\dd t/(2\pi)$, where $r(\cdot)$ 
is the resolvent of $h$. However, in this case we are restricted to 
$\Re z<-2\Sigma$, while considering singular perturbations we deal with 
$R(\pm\img)$. Even if one shows that one can analytically relax the 
restriction, the exposition becomes highly complicated due to the 
hypergeometric origin of functions $G_1$ and $G_2$ in 
\cite{Jursenas14}. On the other hand, we note in remark~\ref{rem:rem1}
without proof how the unitary group computed in the present paper 
relates to $G_1$ and $G_2$.
%%%%%%%%%%%%%%%%%%%%%%%%%%%%%%%%%%%%%%%%%%%%%%%%%%%%%%%%%%%%%%
%%%%%%%%%%%%%%%%%%%%%%%%%%%%%%%%%%%%%%%%%%%%%%%%%%%%%%%%%%%%%%
%%%%%%%%%%%%%%%%%%%%%%%%%%%%%%%%%%%%%%%%%%%%%%%%%%%%%%%%%%%%%%
%%%%%%%%%%%%%%%%%%%%%%%%%%%%%%%%%%%%%%%%%%%%%%%%%%%%%%%%%%%%%%
\section{Preliminaries}
%%%%%%%%%%%%%%%%%%%%%%%%%%%%%%%%%%%%%%%%%%%%%%%%%%%%%%%%%%%%%%
%%%%%%%%%%%%%%%%%%%%%%%%%%%%%%%%%%%%%%%%%%%%%%%%%%%%%%%%%%%%%%
%%%%%%%%%%%%%%%%%%%%%%%%%%%%%%%%%%%%%%%%%%%%%%%%%%%%%%%%%%%%%%
%%%%%%%%%%%%%%%%%%%%%%%%%%%%%%%%%%%%%%%%%%%%%%%%%%%%%%%%%%%%%%
Here and elsewhere: 
\begin{itemize}
\item
$\nabla_i$, $i\in\{1,2,3\}$, is the gradient in the $i$th 
component of a three-dimensional position vector; the Laplacian 
$\Delta=\sum_{i=1}^3\nabla_i^2$. 
\item 
$\Hl=\Hs\ot\C^2$, $\Hs=L^2(\R^3)$, $\Hs^c=L^2(\R^6)$,
$H^m$ ($m\in\N$) is the ($L^2$-)Sobolev space.
\item 
The tensor product $A\ot B=\ol{A\odot B}$ of the operators 
$A$ and $B$ in Hilbert spaces $\Hl_1$ and $\Hl_2$, respectively, 
is the closure of the operator $A\odot B$ defined on the linear 
space spanned by elementary tensors $u\ot v$ (conjugate-bilinear 
forms on $\Hl_1\times\Hl_2$), with $u\in\dom A$ and $v\in\dom B$: 
$(A\odot B)(u\ot v)=Au\ot Bv$. The tensor product $\Hl_1\ot\Hl_2$ 
of Hilbert spaces $\Hl_1$ and $\Hl_2$ is the completion of 
$\Hl_1\odot\Hl_2$ with respect to the cross norm, 
where $\Hl_1\odot\Hl_2$ is the linear space spanned by elementary 
tensors $u\ot v$, with $u\in\Hl_1$ and $v\in\Hl_2$.
\end{itemize}
Denote the standard basis of $\C^2$ by 
\[\ket{\frac{1}{2}}=\begin{pmatrix}1 \\ 0\end{pmatrix}, 
\quad
\ket{-\frac{1}{2}}=\begin{pmatrix}0 \\ 1\end{pmatrix}\,.\]
With this notation
\begin{equation}
S^\pm\ket{s}=\delta_{s,\mp\frac{1}{2}}\ket{\pm\frac{1}{2}}\,, 
\quad
S^3\ket{s}=s\ket{s}
\label{eq:Saction}
\end{equation}
for $s\in\{-\frac{1}{2},\frac{1}{2}\}$; $\delta_{s,\mp\frac{1}{2}}$ 
is the Kronecker symbol. The spin operators $S^\pm$ are bounded 
in $\C^2$, with the adjoint ones $(S^{\pm})^*=S^\mp$. 

By Gauss formula, the adjoint operators $(D^{\pm})^*=-D^\mp$ 
are densely defined in $\Hs$ and hence the $D^\pm=\ol{D^\pm}$ 
are closed. We therefore have that 
$(D^\mp\ot S^\pm)^*=-D^\pm\ot S^\mp$. By definition, 
$\dom(D^\mp\odot S^\pm)$ is dense in $\Hl$, hence so is 
$\dom(D^\mp\ot S^\pm)$. Thus we have
\begin{align*}
(D^-\ot S^+-D^+\ot S^-)^*\supseteq&(D^-\ot S^+)^*-(D^+\ot S^-)^* \\
=&D^-\ot S^+-D^+\ot S^-
\end{align*}
\ie $D^-\ot S^+-D^+\ot S^-$ is densely defined, symmetric, closable. 
In fact, a stronger property holds; see proposition~\ref{prop:ess2}.
%%%
\begin{prop}\label{prop:ess}
Let $\D_b(D^+D^-)$ be the set of bounded vectors $u$ for 
a self-adjoint operator $D^+D^-$; that is, 
$\norm{(D^+D^-)^nu}_\Hs\leq C_u^n$ for some $C_u>0$ and for 
$n\in\N$. Let $\D$ be the set spanned by elementary tensors 
$u\ot\varphi$, where $u\in\D_b(D^+D^-)$ and $\varphi\in\C^2$. 
Then $\D\subseteq\Hl$ densely.
\end{prop}
%%%
\begin{proof}
Applying Gauss formula twice one finds that $D^+D^-$ is self-adjoint. 
Therefore, by standard argument, 
$\D_b(D^+D^-)$ is dense in $\Hs$ and $\D$ is dense in $\Hl$.
\end{proof}
%%%
\begin{prop}\label{prop:ess2}
The operator $D^-\ot S^+-D^+\ot S^-$ is essentially self-adjoint, 
and $\D$ is the core for its closure.
\end{prop}
%%%
\begin{proof}
Consider an elementary tensor $u\ot\varphi\in\D$. 
Each $\varphi\in\C^2$ is of the form $\varphi=\sum_s\varphi_s\ket{s}$, 
with $\varphi_s\in\C$. Let $u_n=(D^+D^-)^nu$. Using \eqref{eq:Saction} 
we have by induction
\begin{align*}
(D^-\ot S^+-D^+\ot S^-)^{2n}(u\ot\varphi)=&
(-1)^nu_n\ot\varphi\,, 
\\
(D^-\ot S^+-D^+\ot S^-)^{2n+1}(u\ot\varphi)=&
(-1)^nD^-u_n\ot \varphi_{-\frac{1}{2}} \ket{\frac{1}{2}} 
\\
&-(-1)^nD^+u_n\ot \varphi_{\frac{1}{2}}\ket{-\frac{1}{2}}
\end{align*}
for $n\in\N_0$. Since $\norm{u_n}_\Hs\leq C_u^n$ 
(see proposition~\ref{prop:ess}), we have
\begin{equation}
\norm{(D^-\ot S^+-D^+\ot S^-)^{2n}(u\ot\varphi)}_\Hl\leq
\sqrt{C_u^{2n}}\norm{\varphi}_{\C^2}\,. 
\label{eq:1t}
\end{equation}
For odd powers we have 
\[\norm{(D^-\ot S^+-D^+\ot S^-)^{2n+1}(u\ot\varphi)}_\Hl\leq
\abs{\varphi_{-\frac{1}{2}}}\,\norm{D^-u_n}_\Hs
+\abs{\varphi_{\frac{1}{2}}}\,\norm{D^+u_n}_\Hs\,.\]
But
\[0\leq\norm{D^\pm u_n}_\Hs^2=
-\braket{u_n,u_{n+1}}_\Hs\leq\norm{u_n}_\Hs\,
\norm{u_{n+1}}_\Hs\leq C_u^{2n+1}\]
and hence
\begin{equation}
\norm{(D^-\ot S^+-D^+\ot S^-)^{2n+1}(u\ot\varphi)}_\Hl
\leq
\sqrt{C_u^{2n+1}}(\abs{\varphi_{-\frac{1}{2}}}
+\abs{\varphi_{\frac{1}{2}}})\,.
\label{eq:2t}
\end{equation}
It follows from \eqref{eq:1t} and \eqref{eq:2t} that 
$u\ot\varphi\in\D$ is a bounded, hence analytic, vector for 
$D^-\ot S^+-D^+\ot S^-$. Since $\D$ is dense in $\Hl$ by 
proposition~\ref{prop:ess}, the set of analytic vectors for 
a symmetric operator $D^-\ot S^+-D^+\ot S^-$ is also dense in $\Hl$. 
Then, by Nelson theorem, $D^-\ot S^+-D^+\ot S^-$ is essentially 
self-adjoint.
\end{proof}
%%%
By proposition~\ref{prop:ess2}, the operator $h$ on 
$\dom h=\dom h^0$ is self-adjoint.
%%%%%%%%%%%%%%%%%%%%%%%%%%%%%%%%%%%%%%%%%%%%%%%%%%%%%%%%%%%%%%
%%%%%%%%%%%%%%%%%%%%%%%%%%%%%%%%%%%%%%%%%%%%%%%%%%%%%%%%%%%%%%
%%%%%%%%%%%%%%%%%%%%%%%%%%%%%%%%%%%%%%%%%%%%%%%%%%%%%%%%%%%%%%
%%%%%%%%%%%%%%%%%%%%%%%%%%%%%%%%%%%%%%%%%%%%%%%%%%%%%%%%%%%%%%
\section{Integral kernel of the unitary group}
%%%%%%%%%%%%%%%%%%%%%%%%%%%%%%%%%%%%%%%%%%%%%%%%%%%%%%%%%%%%%%
%%%%%%%%%%%%%%%%%%%%%%%%%%%%%%%%%%%%%%%%%%%%%%%%%%%%%%%%%%%%%%
%%%%%%%%%%%%%%%%%%%%%%%%%%%%%%%%%%%%%%%%%%%%%%%%%%%%%%%%%%%%%%
%%%%%%%%%%%%%%%%%%%%%%%%%%%%%%%%%%%%%%%%%%%%%%%%%%%%%%%%%%%%%%
Recall that the confluent Humbert function $\Phi_3$ possesses 
the series representation \cite[equation~(40)]{Srivastava} 
\[\Phi_3(a;b;x,y)=\sum_{m,n\in\N_0}
\frac{(a)_m}{(b)_{m+n}}\frac{x^m}{m!}
\frac{y^n}{n!}\]
which is absolutely convergent for $\abs{x},\abs{y}<\infty$, 
provided that the Pochhammer symbols exist (\ie $-b\notin\N_0$). 
The $\Phi_3$ function is one of the seven confluent forms of the 
Appell $F_1$ function; see also \cite{Debiard}.
%%%
\begin{thm}\label{thm:lE}
For $t>0$ and $u\ot\ket{s}\in\D$
\begin{equation}
\e^{-th}(u\ot\ket{s})=\sum_{s^\prime}
(G_t^{s^\prime s}*u)\ot\ket{s^\prime}
\label{eq:kerh-1}
\end{equation}
where 
\begin{subequations}\label{eq:kerh-2}
\begin{equation}
G_t^{s^\prime s}(x)=
K_t^0(x)[\delta_{s^\prime s}(a_t-2\beta s b_t)
+\frac{\alpha}{2t}b_t
(\delta_{s^\prime,\frac{1}{2}}\delta_{s,-\frac{1}{2}}x^-
-\delta_{s^\prime,-\frac{1}{2}}\delta_{s,\frac{1}{2}}x^+)]\,, 
\end{equation}
\begin{equation}
a_t=\Phi_3(1;1/2;t(\alpha/2)^2,(\beta t/2)^2)\,, \quad 
b_t=t\Phi_3(1;3/2;t(\alpha/2)^2,(\beta t/2)^2)
\end{equation}
\end{subequations}
for a.e. $x=(x^1,x^2,x^3)\in\R^3$; $x^\pm=x^1\pm\img x^2$.
\end{thm}
%%%
\begin{cor}\label{cor:lE}
For $t\in\R\setm\{0\}$ and $u\ot\ket{s}\in\D$
\begin{subequations}\label{eq:ker}
\begin{align}
\e^{-\img th}(u\ot\ket{s})=&
\slim_{\epsilon\searrow0}
\sum_{s^\prime}(G_{\epsilon+\img t}^{s^\prime s}*u)
\ot\ket{s^\prime} \\
=&
\sum_{s^\prime}\LIM\int
G_{\img t}^{s^\prime s}(\cdot-y)u(y)\dd y\ot\ket{s^\prime}\,.
\end{align}
\end{subequations}
\end{cor}
%%%
\begin{rem}\label{rem:rem1}
1) Let $U_t=\e^{\img th}\vert_\D$. 
Then $U_t\subseteq \e^{\img th}\Rightarrow
U_t^*=\e^{-\img th}\Rightarrow\ol{U_t}=\e^{\img th}$. 
Thus, the extension $\e^{\img th}$ of $U_t$ acts on $u\in\Hs$ 
by \eqref{eq:ker}: 
For all $\epsilon>0$ and a given $f\in\Hl$ one finds a $g\in\D$ such that 
$\norm{f-g}_\Hl<\epsilon$, so $\norm{\e^{\img th}f}_\Hl\leq
\norm{U_tg}_\Hl+\epsilon$.

2) Using $(h-z\id_\Hl)^{-1}=\int_0^\infty\e^{tz}\e^{-th}\dd t$ 
for $\Re z<-\Sigma$, one shows that
\begin{subequations}\label{eq:cgrtsf}
\begin{align}
G_1(x;z)=&\int_0^\infty\e^{tz}K_t^0(x)b_t \dd t\,, 
\quad x\in\R^3\,,  
\\
D^\pm G_1(x;z)=&
-\frac{x^\pm}{2}\int_0^\infty\e^{tz}K_t^0(x)b_t t^{-1} 
\dd t\,, \quad x\in\R^3\setm\{0\}\,,  
\\
G_2(x;z)=&\int_0^\infty\e^{tz}K_t^0(x)a_t \dd t\,, 
\quad x\in\R^3\setm\{0\}
\end{align}
\end{subequations}
provided that the hypergeometric series on the left
exist (theorem~3.3 in \cite{Jursenas14}). Initially, 
relations in \eqref{eq:cgrtsf} are valid for a.e. $x$, 
as it is seen from the comparison between the single-particle
Green function in \cite{Jursenas14} and that derived by using 
\eqref{eq:kerh-1}. However, the direct computation of the 
integrals on the right shows that the equalities 
hold true for all $x$. One therefore refers to \eqref{eq:cgrtsf} 
as the integral representations of the hypergeometric series 
$G_1$, $D^\pm G_1$, $G_2$.

3) As $\alpha\searrow0$
\begin{subequations}\label{eq:atbt-small}
\begin{align}
a_{\img t}=&\cos(\beta t)
+\frac{2\img}{\beta}(\alpha/2)^2\sin(\beta t)+O(\alpha^4)\,, 
\\
b_{\img t}=&\frac{\img}{\beta}\sin(\beta t)
+\frac{2}{\beta^2}(\alpha/2)^2(\cos(\beta t)
-\frac{\sin(\beta t)}{\beta t})+O(\alpha^4)\,.
\end{align}
\end{subequations}
If further $\beta=0$, one assumes the limit $\beta\searrow0$. 
Similar relations hold for $\alpha$ arbitrary and $\beta$ small. 
We use \eqref{eq:atbt-small} for computing the two-particle 
Green function later on.
\end{rem}
%%%

Now we give the proofs of the above results.
The proof of theorem~\ref{thm:lE} relies on the two lemmas.
%%%
\begin{lem}\label{lem:lE1}
For $t>0$, the operator
\begin{equation}
\e^{t\Delta\vert_{H^2(\R^3)}}
\cosh(t\sqrt{-\alpha^2D^+D^-+\beta^2\id_\Hs})
\vert_{\D_b(D^+D^-)}
\label{eq:bbb1}
\end{equation}
in $\Hs$ admits the unique extension which is bounded in $\Hs$. 
The integral kernel $k_t$ of the operator in \eqref{eq:bbb1} 
is given by
\[k_t(x)=K_t^0(x)\Phi_3(1;1/2;t(\alpha/2)^2,(\beta t/2)^2)\]
for a.e. $x\in \R^3$.
\end{lem}
%%%
\begin{lem}\label{lem:lE2}
For $t>0$, the operator
\begin{equation}
\e^{t\Delta\vert_{H^2(\R^3)}}
\frac{\sinh(t\sqrt{-\alpha^2D^+D^-+\beta^2\id_\Hs})}{
\sqrt{-\alpha^2D^+D^-+\beta^2\id_\Hs}}\vert_{\D_b(D^+D^-)}
\label{eq:bbb2}
\end{equation}
in $\Hs$ admits the unique extension which is bounded in $\Hs$. 
The integral kernel $l_t$ of the operator in \eqref{eq:bbb2} 
is given by
\[l_t(x)=K_t^0(x)t\Phi_3(1;3/2;t(\alpha/2)^2,(\beta t/2)^2)\]
for a.e. $x\in \R^3$.
\end{lem}
%%%
\begin{proof}[Proof of lemma~\ref{lem:lE1}]
Throughout, $\norm{\cdot}_p$ is the norm in $L^p(\R^3)$ 
for $1\leq p<\infty$; in particular $\norm{\cdot}_2=\norm{\cdot}_\Hs$.
Let $p_t(-\img\nabla)$ denote the operator in \eqref{eq:bbb1}. 
First we show that $p_t(-\img\nabla)$ defines a continuous mapping 
$\D_b(D^+D^-)\lto L^2(\R^3)$.

Using $\cosh(\cdot)=\slim_{n\nearrow\infty}
\sum_{k=0}^n(\cdot)^{2k}/(2k)!$
on $\D_b(D^+D^-)$, and the binomial formula, we have for 
$u\in\D_b(D^+D^-)$ 
\begin{align*}
\norm{p_t(-\img\nabla)u}_2=&
\norm{\sum_{n\in\N_0}\frac{(\beta t)^{2n}}{(2n)!}
\sum_{m=0}^n\binom{n}{m}(-1)^m
(\frac{\alpha}{\beta})^{2m}\e^{t\Delta}u_m}_2 
\\
\leq&\sum_{n\in\N_0}\frac{(\beta t)^{2n}}{(2n)!}
\sum_{m=0}^n\binom{n}{m}
(\frac{\alpha}{\beta})^{2m}\norm{\e^{t\Delta}u_m}_2
\end{align*}
where $u_m=(D^+D^-)^mu$. Since $K_t^0\in L^1(\R^3)$, 
we have by Young inequality
\begin{equation}
\norm{\e^{t\Delta}u_m}_2=
\norm{K_t^0*u_m}_2\leq \norm{K_t^0}_1\,\norm{u_m}_2\,.
\label{eq:youngu}
\end{equation}
But $\norm{u_m}_2\leq C_u^m$ for some finite $C_u>0$, 
and $\norm{K_t^0}_1=1$, so 
\begin{align*}
\norm{p_t(-\img\nabla)u}_2\leq& 
\sum_{n\in\N_0}\frac{(\beta t)^{2n}}{(2n)!}
\sum_{m=0}^n\binom{n}{m}
(\frac{\alpha}{\beta})^{2m}C_u^m 
=\cosh(t\sqrt{\alpha^2C_u+\beta^2})
\end{align*}
showing that $p_t(-\img\nabla)$ is a continuous linear 
operator from $\D_b(D^+D^-)$ into $L^2(\R^3)$. Since 
$\D_b(D^+D^-)$ is dense in $L^2(\R^3)$, there exists the 
unique bounded operator in $L^2(\R^3)$ that extends  
$p_t(-\img\nabla)$.

Next we compute $k_t$. Since $p_t(-\img\nabla)u=k_t*u$ for 
$u\in\D_b(D^+D^-)$, we have
\begin{equation}
k_t=p_t(-\img\nabla)\delta
\label{eq:KtE}
\end{equation} 
in $\Ds^\prime(\R^3)$, where $\delta$ is the Dirac distribution.
Let $\F_{x\lmap\xi}\co\Hs\lto\Hs$ be the Fourier transform; 
$\xi=(\xi^1,\xi^2,\xi^3)\in\R^3$. In the sense of 
Plancherel--Riesz theorem, for $u\in \Hs$
\begin{equation}
\F[D^\pm u]=\img\xi^\pm\hat{u}\,, \quad 
\xi^\pm=\xi^1\pm\img\xi^2\,, \quad 
\hat{u}=\F[u]\,.
\label{eq:fE4}
\end{equation}
Using \eqref{eq:fE4}, $p_t(-\img\nabla)$ has the symbol
\begin{equation}
p_t(\xi)=\e^{-t\xi\cdot\xi}
\cosh(t\sqrt{\alpha^2\xi^+\xi^-+\beta^2})\,.
\label{eq:fE5-0}
\end{equation}
The dot product denotes the standard scalar product 
of vectors in $\R^3$. Since the symbol $p_t\in L^1(\R^3)$, 
relation \eqref{eq:KtE} gives
\begin{equation}
k_t(x)=\frac{1}{(2\pi)^3}\int_{\R^3}
\e^{\img x\cdot\xi}p_t(\xi)\dd\xi \quad
\text{a.e. $x\in\R^3$\,.}
\label{eq:fE5}
\end{equation}
Rewrite $\xi$ in \eqref{eq:fE5-0} and \eqref{eq:fE5} 
in spherical coordinates $(\abs{\xi}=\rho,\theta,\phi)$ 
and use the series representation of $\cosh$. Then, expand 
$(\alpha^2\xi^+\xi^-+\beta^2)^n$ using binomial formula, 
apply the relations 
\begin{subequations}\label{eq:ints}
\begin{align}
&\int_0^\pi\e^{\img \abs{x}\rho\cos\theta}
\sin^{2m+1}\theta \dd\theta
=\frac{\sqrt{\pi}\,m!}{\Gamma(m+\frac{3}{2})}
\,_0F_1(;m+\frac{3}{2};-\frac{\abs{x}^2\rho^2}{4})\,, \\
&\int_0^\infty\rho^{m+1/2}\e^{-t\rho}\,
_0F_1(;m+\frac{3}{2};-\frac{\abs{x}^2\rho}{4})\dd\rho
=\e^{-\frac{\abs{x}^2}{4t}}t^{-m-3/2}\Gamma(m+\frac{3}{2})
\end{align}
\end{subequations}
($m\in\N_0$), and get that
\[k_t(x)=\frac{\e^{-\frac{\abs{x}^2}{4t}}a_t}{(4\pi t)^{3/2}}\,, 
\quad
a_t=\sum_{n\in\N_0}\frac{(\beta t)^{2n}}{(2n)!}\sum_{m=0}^n
\binom{n}{m}(\frac{\alpha}{\beta})^{2m}m!t^{-m}\,.\]

It remains to compute $a_t$. The formal double series
\[\sum_{n\in\N_0}\sum_{m=0}^n\Omega_{nm}
=\sum_{n\in\N_0}\sum_{m=n}^\infty \Omega_{mn}\]
provided that the sum exists. Using in addition that 
$1/\Gamma(k)=0$ for $-k\in\N_0$, we get 
\begin{align*}
a_t=&\sum_{n,m\in\N_0}\frac{(\beta t)^{2m}}{(2m)!}
(\frac{\alpha}{\beta})^{2n}t^{-n}\frac{m!}{\Gamma(m-n+1)} \\
=&\sum_{n,m\in\N_0}\frac{[(\alpha/\beta)^2/t]^n}{n!}
\frac{[(\beta t)^2]^m}{m!}
\left[\frac{(1)_n(1)_m(1)_m}{(1)_{2m}(1)_{m-n}}\right]\,.
\end{align*}
Using Legendre's duplication formula 
$(1)_{2m}=4^m(\frac{1}{2})_m(1)_m$ and 
then replacing $m$ by $n+p$ for $p\in\N_0$ 
(since $1/(1)_p=0$ for $-p\in\N$) we get
\begin{align*}
a_t=&
\sum_{n,p\in\N_0}\frac{[(\alpha/\beta)^2/t]^n}{n!}
\frac{[(\beta t)^2]^{n+p}}{(n+p)!}
\left[\frac{4^{-n-p}(1)_{n+p}(1)_n}{
(\frac{1}{2})_{n+p}(1)_p}\right] \\
=&\sum_{n,p\in\N_0}\frac{[t(\alpha/2)^2]^n}{n!}
\frac{[(\beta t/2)^2]^{p}}{p!}
\left[\frac{(1)_n}{(\frac{1}{2})_{n+p}}\right]
\end{align*}
which is the confluent Humbert function 
$\Phi_3(1;1/2;t(\alpha/2)^2,(\beta t/2)^2)$.
This completes the proof of the lemma.
\end{proof}
%%%
\begin{proof}[Proof of lemma~\ref{lem:lE2}]
Let $p_t(-\img\nabla)$ denote the operator in \eqref{eq:bbb2}. 
First we show that $p_t(-\img\nabla)$ defines a continuous mapping 
$\D_{b}(D^+D^-)\lto L^2(\R^3)$. 

Using $\sinh(\cdot)=\slim_{n\nearrow\infty}
\sum_{k=0}^n(\cdot)^{2k+1}/(2k+1)!$ on $\D_b(D^+D^-)$, 
and the binomial formula, we have for $u\in\D_b(D^+D^-)$ 
\begin{align*}
\norm{p_t(-\img\nabla)u}_2=&
\norm{\sum_{n\in\N_0}\frac{(\beta t)^{2n}t}{
(2n+1)!}\sum_{m=0}^n\binom{n}{m}(-1)^m
(\frac{\alpha}{\beta})^{2m}\e^{t\Delta}u_m}_2 \\
\leq&t\sum_{n\in\N_0}\frac{(\beta t)^{2n}}{(2n+1)!}
\sum_{m=0}^n\binom{n}{m}
(\frac{\alpha}{\beta})^{2m}\norm{\e^{t\Delta}u_m}_2
\end{align*}
where $u_m=(D^+D^-)^mu$. It was shown in the proof 
of lemma~\ref{lem:lE1} that $\norm{\e^{t\Delta}u_m}_2\leq C_u^m$ 
for some finite $C_u>0$, so 
\[\norm{p_t(-\img\nabla)u}_2\leq 
\frac{\sinh(t\sqrt{\alpha^2C_u+\beta^2})}{
\sqrt{\alpha^2C_u+\beta^2}}\]
implying that $p_t(-\img\nabla)$ is a continuous linear operator from 
$\D_b(D^+D^-)$ into $L^2(\R^3)$. Since $\D_b(D^+D^-)$ is dense in 
$L^2(\R^3)$, there exists the unique bounded operator in $L^2(\R^3)$ 
that extends $p_t(-\img\nabla)$.

The computation of $l_t$ is pretty much the same as that 
of $k_t$ in lemma~\ref{lem:lE1}. In this case the symbol of 
$p_t(-\img\nabla)$ is 
\[p_t(\xi)=\e^{-t\xi\cdot\xi}
\frac{\sinh(t\sqrt{\alpha^2\xi^+\xi^-+\beta^2})}{
\sqrt{\alpha^2\xi^+\xi^-+\beta^2}}\,,
\quad p_t\in L^1(\R^3)\,.\]
Using the series representation of $\sinh$, 
\eqref{eq:fE5}, and \eqref{eq:ints}, we get
\[l_t(x)=\frac{\e^{-\frac{\abs{x}^2}{4t}}b_t}{(4\pi t)^{3/2}}\,, 
\quad 
b_t=t\sum_{n\in\N_0}\frac{(\beta t)^{2n}}{(2n+1)!}\sum_{m=0}^n
\binom{n}{m}(\frac{\alpha}{\beta})^{2m}m!t^{-m}\,.\]
Proceeding exactly the same way as when equating $a_t$ in 
the proof of lemma~\ref{lem:lE1}, we get 
$b_t=t\Phi_3(1;3/2;t(\alpha/2)^2,(\beta t/2)^2)$
(the only difference is that now 
$(2m+1)!=(2)_{2m}=4^m(\frac{3}{2})_m(1)_m$).
The proof of the lemma is accomplished.
\end{proof}
%%%
We are in a position to accomplish the proof of the theorem.
%%%
\begin{proof}[Proof of theorem~\ref{thm:lE}]
Both $\e^{-th^0}\vert_\D$ and $\e^{-tU}\vert_\D$ are 
the strong limits of their Taylor series. Since $\D$ is 
invariant under the action of $\e^{-th^0}$ 
(because of \eqref{eq:youngu}) and $\e^{-tU}$ (by definition),
the operators $\e^{-th^0}$ and $\e^{-tU}$ commute on $\D$, 
and it holds
\begin{equation}
\e^{-th}\vert_\D=\e^{-th^0}\e^{-tU}\vert_\D\,.
\label{eq:e-1}
\end{equation}
Next, using
\[(S^\pm)^n=0\quad(n\in\N_{\geq2})\,, \quad 
S^+S^-+S^-S^+=\id_{\C^2}\,, \quad
S^\pm S^3+S^3S^\pm=0\]
we have by induction 
\[U^{2n}(u\ot\ket{s})
=(-\alpha^2D^+D^-+\beta^2\id_\Hs)^nu\ot\ket{s}\]
for $n\in\N_0$ and $u\ot\ket{s}\in\D$. 
Therefore, by \eqref{eq:e-1}
\begin{align*}
\e^{-th}(u\ot\ket{s})=&
(\e^{t\Delta\vert_{H^2(\R^3)}}\odot\id_{\C^2})
\sum_{n\in\N_0}\frac{(-t)^n}{n!}U^n(u\ot\ket{s}) 
\\
=&\sum_{n\in\N_0}\frac{t^{2n}}{(2n)!}
\e^{t\Delta\vert_{H^2(\R^3)}}
(-\alpha^2D^+D^-+\beta^2\id_\Hs)^nu\ot\ket{s} 
\\
&+\sum_{n\in\N_0}\frac{(-t)^{2n+1}}{(2n+1)!} 
[\e^{t\Delta\vert_{H^2(\R^3)}}(-\alpha^2D^+D^-+
\beta^2\id_\Hs)^n\odot\id_{\C^2}]U(u\ot\ket{s})\,.
\end{align*}
The first series is $[\text{equation~}\eqref{eq:bbb1}]u\ot\ket{s}$ 
and the second one is 
$-([\text{equation~}\eqref{eq:bbb2}]\odot\id_{\C^2})U(u\ot\ket{s})$.
Using the definition of $U$, and then applying
lemmas~\ref{lem:lE1} and \ref{lem:lE2}, and 
\[D^\pm l_t(x)=-\frac{x^\pm}{2t}l_t(x) \quad 
\text{a.e. $x\in\R^3$}\]
one deduces \eqref{eq:kerh-1} and \eqref{eq:kerh-2}. 
The proof of the theorem is complete.
\end{proof}
%%%
\begin{proof}[Proof of corollary~\ref{cor:lE}]
For $t\in\R\setm\{0\}$, $\epsilon>0$, $u\ot\ket{s}\in\D$, 
we have by theorem~\ref{thm:lE}
\[\norm{(\e^{-\img(t-\img\epsilon)h}-\e^{-\img th})
(u\ot\ket{s})}_\Hl
=\norm{\sum_{s^\prime}\varphi_\epsilon^{s^\prime s}
\ot\ket{s^\prime}}_\Hl
=\left(\sum_{s^\prime}
\norm{\varphi_\epsilon^{s^\prime s}}_\Hs^2\right)^{1/2}\]
where 
\[\varphi_\epsilon^{s^\prime s}=
G_\epsilon^{s^\prime s}*u-\delta_{s^\prime s}u\,.\]
Further, using $\lim_{t\lto0}a_t=1$ and $\lim_{t\lto0}b_t=0$
\[\lim_{\epsilon\searrow0}
\norm{\varphi_\epsilon^{s^\prime s}}_\Hs
=\delta_{s^\prime s}\lim_{\epsilon\searrow0}
\norm{K_\epsilon^0*u-u}_\Hs\,.\]
But
\begin{align*}
\lim_{\epsilon\searrow0}\norm{K_\epsilon^0*u-u}_\Hs
=&\lim_{\epsilon\searrow0}
\norm{\e^{\epsilon\Delta}u-u}_\Hs
=\lim_{\epsilon\searrow0}
\epsilon\norm{\frac{1}{\epsilon}
(\e^{\epsilon\Delta\vert_{H^2(\R^3)}}
-\id_\Hs)u}_\Hs 
\\
=&\lim_{\epsilon\searrow0}\epsilon\norm{\Delta u}_\Hs=0
\end{align*}
and hence
\[\e^{-\img th}(u\ot\ket{s})=\slim_{\epsilon\searrow0}
\e^{-\img(t-\img\epsilon)h}(u\ot\ket{s})=
\slim_{\epsilon\searrow0}
\sum_{s^\prime}
(G_{\img(t-\img\epsilon)}^{s^\prime s}*u)\ot\ket{s^\prime}\,.\]
Now, by standard argument, instead of $u$, consider the function 
$u_R=1_{\{x\vrt\abs{x}\leq R\}}u$. Then
$u_R\in\D_b(D^+D^-)\cap L^1(\R^3)$ and one deduces \eqref{eq:ker}.
\end{proof}
%%%%%%%%%%%%%%%%%%%%%%%%%%%%%%%%%%%%%%%%%%%%%%%%%%%%%%%%%%%%%%
%%%%%%%%%%%%%%%%%%%%%%%%%%%%%%%%%%%%%%%%%%%%%%%%%%%%%%%%%%%%%%
%%%%%%%%%%%%%%%%%%%%%%%%%%%%%%%%%%%%%%%%%%%%%%%%%%%%%%%%%%%%%%
%%%%%%%%%%%%%%%%%%%%%%%%%%%%%%%%%%%%%%%%%%%%%%%%%%%%%%%%%%%%%%
\section{Green function for the two-particle operator}
%%%%%%%%%%%%%%%%%%%%%%%%%%%%%%%%%%%%%%%%%%%%%%%%%%%%%%%%%%%%%%
%%%%%%%%%%%%%%%%%%%%%%%%%%%%%%%%%%%%%%%%%%%%%%%%%%%%%%%%%%%%%%
%%%%%%%%%%%%%%%%%%%%%%%%%%%%%%%%%%%%%%%%%%%%%%%%%%%%%%%%%%%%%%
%%%%%%%%%%%%%%%%%%%%%%%%%%%%%%%%%%%%%%%%%%%%%%%%%%%%%%%%%%%%%%
In this section we study the resolvent of the two-particle 
operator $H$ in \eqref{eq:H0} written in the center-of-mass 
coordinate system.
%%%%%%%%%%%%%%%%%%%%%%%%%%%%%%%%%%%%%%%%%%%%%%%%%%%%%%%%%%%%%%
%%%%%%%%%%%%%%%%%%%%%%%%%%%%%%%%%%%%%%%%%%%%%%%%%%%%%%%%%%%%%%
%%%%%%%%%%%%%%%%%%%%%%%%%%%%%%%%%%%%%%%%%%%%%%%%%%%%%%%%%%%%%%
%%%%%%%%%%%%%%%%%%%%%%%%%%%%%%%%%%%%%%%%%%%%%%%%%%%%%%%%%%%%%%
\subsection{Bases}
%%%%%%%%%%%%%%%%%%%%%%%%%%%%%%%%%%%%%%%%%%%%%%%%%%%%%%%%%%%%%%
%%%%%%%%%%%%%%%%%%%%%%%%%%%%%%%%%%%%%%%%%%%%%%%%%%%%%%%%%%%%%%
%%%%%%%%%%%%%%%%%%%%%%%%%%%%%%%%%%%%%%%%%%%%%%%%%%%%%%%%%%%%%%
%%%%%%%%%%%%%%%%%%%%%%%%%%%%%%%%%%%%%%%%%%%%%%%%%%%%%%%%%%%%%%
In what follows we find it convenient to identify $\C^2\ot\C^2$ with 
$\C^3\op\C^1$, which basically says that the tensor product 
$\mathbf{2}\ot\mathbf{2}$ of $SU_2$-irreducible 
representations of dimensions $2s+1=2$ reduces to the orthogonal sum 
$\mathbf{3}\op\mathbf{1}$ of $SU_2$-irreducible representations of 
dimensions $2s+1=3$ and $2s+1=1$. 
The basis of the space of the representation $\mathbf{3}$ is 
$(\ket{1s})_{s\in\{-1,0,1\}}$, while 
the space of $\mathbf{1}$ is single-dimensional, with 
the basis $\ket{00}$. This follows from the fact that the basis of 
$\C^3\op\C^1$
\[(\ket{\sigma})_{\sigma\in\bas}\,, \quad
\bas=\bigcup_{S\in\{0,1\}}\{(S,s)\vrt 
s\in\{-S,-S+1,\ldots,S\}\}\]
is orthonormal and is related to the orthonormal basis
\[(\ket{s_1s_2}=\ket{s_1}\ot\ket{s_2})_{(s_1,s_2)\in
\{-\frac{1}{2},\frac{1}{2}\}^2}\]
of $\C^2\ot\C^2$ via the Clebsch--Gordan coefficient of 
$SU_2$ \cite{Jucys}:
\begin{equation}
\ket{Ss}=\sum_{s_1,s_2}
\cgc{\frac{1}{2}}{\frac{1}{2}}{S}{s_1}{s_2}{s}\ket{s_1s_2}
\label{eq:cgcrel}
\end{equation}
where the sum runs over 
$(s_1,s_2)\in\{-\frac{1}{2},\frac{1}{2}\}^2$ 
such that $s_1+s_2=s$. Using \eqref{eq:cgcrel} and the 
orthogonality condition for the Clebsch--Gordan coefficient, 
one represents the basis vector $\ket{s_1s_2}$ of $\C^2\ot\C^2$ 
as a linear combination of the basis vectors $\ket{\sigma}$ 
of $\C^3\op\C^1$, with $\sigma$ ranging over $\bas$.

Given the orthonormal basis $(e_i)$ of $\Hs$, 
the orthonormal basis of $\Hl$ is $(e_i\ot\ket{s})$ 
and the orthonormal basis of $\Hl\ot\Hl$ is 
$((e_i\ot\ket{s_1})\ot (e_j\ot\ket{s_2}))$. The Hilbert space 
\[\K=(\Hs\ot\Hs)\ot(\C^3\op\C^1)\]
has orthonormal basis $(e_{ij}\ot\ket{\sigma})$, 
with $e_{ij}=e_i\ot e_j$, $\sigma=(S,s)\in\bas$. The map 
$J\co \Hl\ot\Hl\lto\K$ given by
\[J\co (e_i\ot\ket{s_1})\ot (e_j\ot\ket{s_2})\lmap
\sum_\sigma
\cgc{\frac{1}{2}}{\frac{1}{2}}{S}{s_1}{s_2}{s}
e_{ij}\ot\ket{\sigma}\]
is unitary, and so $\K$ is isomorphic to $\Hl\ot\Hl$.
%%%%%%%%%%%%%%%%%%%%%%%%%%%%%%%%%%%%%%%%%%%%%%%%%%%%%%%%%%%%%%
%%%%%%%%%%%%%%%%%%%%%%%%%%%%%%%%%%%%%%%%%%%%%%%%%%%%%%%%%%%%%%
%%%%%%%%%%%%%%%%%%%%%%%%%%%%%%%%%%%%%%%%%%%%%%%%%%%%%%%%%%%%%%
%%%%%%%%%%%%%%%%%%%%%%%%%%%%%%%%%%%%%%%%%%%%%%%%%%%%%%%%%%%%%%
\subsection{Center-of-mass coordinate system}
%%%%%%%%%%%%%%%%%%%%%%%%%%%%%%%%%%%%%%%%%%%%%%%%%%%%%%%%%%%%%%
%%%%%%%%%%%%%%%%%%%%%%%%%%%%%%%%%%%%%%%%%%%%%%%%%%%%%%%%%%%%%%
%%%%%%%%%%%%%%%%%%%%%%%%%%%%%%%%%%%%%%%%%%%%%%%%%%%%%%%%%%%%%%
%%%%%%%%%%%%%%%%%%%%%%%%%%%%%%%%%%%%%%%%%%%%%%%%%%%%%%%%%%%%%%
Let $x_1\in\R^3$, $x_2\in\R^3$ be the position-vectors 
of the two atoms, and put $q=(x_1,x_2)$ and $Q=(x,X)$, 
where $x=x_1-x_2$ and $X=(x_1+x_2)/2$. Then
\[Kq=Q\,, \quad 
K=\begin{pmatrix}
1 & -1 \\ \frac{1}{2} &
\frac{1}{2}
\end{pmatrix}\,.\]
The coordinate transformation $K$ gives rise to the 
unitary transformation $\mathcal{U}$ in the Hilbert space 
$\Hs^c=L^2(\R^6)$:
\[\mathcal{U}\co\Hs^c\lto\Hs^c\,, \quad 
f(q)\lmap f(Kq)\,.\]
Now consider the unitary isomorphisms 
$\tau\co\Hs\ot\Hs\lto\Hs^c$, $e_{ij}\lmap e_i(x_1)e_j(x_2)$, and 
\[\tilde{\tau}=\mathcal{U}\tau\ot\id_{\C^3\op\C^1}\co 
\K\lto\K^c=\Hs^c\ot(\C^3\op\C^1)\,.\]
Note that $\mathcal{U}\tau\in\B(\Hs\ot\Hs,\Hs^c)$ is bounded, so 
$(\mathcal{U}\tau\ot\id_{\C^3\op\C^1})^*=
(\mathcal{U}\tau)^*\ot\id_{\C^3\op\C^1}$. Define 
\begin{equation}
R^c(z)=LR(z)L^*\,, \quad 
z\in\res H=\C\setm[-2\Sigma,\infty)
\label{eq:Rcz}
\end{equation}
where the unitary map
\begin{align*}
&L=\tilde{\tau}J\co\Hl\ot\Hl\lto\K^c\,, 
\\
&[(e_i\ot\ket{s_1})\ot (e_j\ot\ket{s_2})]\lmap
\sum_{S,s}\cgc{\frac{1}{2}}{\frac{1}{2}}{S}{s_1}{s_2}{s}
e_i(x)e_j(X)\ot\ket{Ss}\,.
\end{align*}
Note that $J\in\B(\Hl\ot\Hl,\K)$, 
so $(\tilde{\tau}J)^*=J^*\tilde{\tau}^*$.
The bounded operator $R^c(z)\in\B(\K^c)$ is the resolvent 
$(H^c-z\id_{\K^c})^{-1}$ of the self-adjoint operator
\[H^c=LHL^*\,, \quad \dom H^c=L\dom H\]
which represents the two-particle Hamiltonian \eqref{eq:H0} 
written in the center-of-mass coordinate system $Q$. 
The operator $H$ is unitarily equivalent to
\eqref{eq:Op-ABD}, and the unitary self-map is given by 
$J^*[\tau^*\mathcal{U}\tau\ot\id_{\C^3\op\C^1}]J$.
%%%%%%%%%%%%%%%%%%%%%%%%%%%%%%%%%%%%%%%%%%%%%%%%%%%%%%%%%%%%%%
%%%%%%%%%%%%%%%%%%%%%%%%%%%%%%%%%%%%%%%%%%%%%%%%%%%%%%%%%%%%%%
%%%%%%%%%%%%%%%%%%%%%%%%%%%%%%%%%%%%%%%%%%%%%%%%%%%%%%%%%%%%%%
%%%%%%%%%%%%%%%%%%%%%%%%%%%%%%%%%%%%%%%%%%%%%%%%%%%%%%%%%%%%%%
\subsection{Resolvent}
%%%%%%%%%%%%%%%%%%%%%%%%%%%%%%%%%%%%%%%%%%%%%%%%%%%%%%%%%%%%%%
%%%%%%%%%%%%%%%%%%%%%%%%%%%%%%%%%%%%%%%%%%%%%%%%%%%%%%%%%%%%%%
%%%%%%%%%%%%%%%%%%%%%%%%%%%%%%%%%%%%%%%%%%%%%%%%%%%%%%%%%%%%%%
%%%%%%%%%%%%%%%%%%%%%%%%%%%%%%%%%%%%%%%%%%%%%%%%%%%%%%%%%%%%%%
By \eqref{eq:resT-2}, \eqref{eq:ker}, and \eqref{eq:Rcz}, 
the resolvent $R^c(z)$ acts on 
$f\ot\ket{\sigma}\in\Hs^c\ot(\C^3\op\C^1)$ as follows:
\begin{equation}
R^c(z)(f(Q)\ot\ket{\sigma})=\sum_{\sigma^\prime}
\LIM\int R_{\sigma^\prime\sigma}^c(z)(Q-Q^\prime)
f(Q^\prime)\dd Q^\prime\ot\ket{\sigma^\prime}\,.
\label{eq:Hresz-c-2}
\end{equation}
For a.e. $Q=(x,X)\in\R^6$, the element 
$R_{\sigma^\prime\sigma}^c(z)(Q)$ 
of the Green function is given by the improper Riemann integral 
($\int_0^\infty\equiv\lim_{\epsilon\searrow0}\int_\epsilon^\infty$)
\begin{equation}
R_{\sigma^\prime\sigma}^c(z)(x,X)
=\pm\img\int_0^\infty\e^{\pm\img tz}
\sum_{\substack{s_1,s_2 \\ s_1^\prime,s_2^\prime}}
\cgc{\frac{1}{2}}{\frac{1}{2}}{S}{s_1}{s_2}{s}
\cgc{\frac{1}{2}}{\frac{1}{2}}{
S^\prime}{s_1^\prime}{s_2^\prime}{s^\prime} 
G_{\pm\img t}^{s_1^\prime s_1}(x)
G_{\pm\img t}^{s_2^\prime s_2}(X)\dd t
\label{eq:HresProj-c-2}
\end{equation}
with the upper (lower) sign taken for $\Im z>0$ ($\Im z<0$). 
Relations~\eqref{eq:Hresz-c-2} and \eqref{eq:HresProj-c-2} 
follow from the fact that the integrand $(\cdot)$ in 
\eqref{eq:HresProj-c-2} is not Lebesgue summable,
since $(\cdot)$ is of the form 
$\exp(\pm\img [tz+Q^2/(4t)])
t^{-3}P_{\pm\img t}^{\sigma^\prime\sigma}(Q)$ with
some $P_{\pm\img t}^{\sigma^\prime\sigma}$ containing the powers 
$t^n$ for $n\in\Z$. Therefore, in order to apply the Fubini 
theorem for changing the order of integration with respect 
to $Q^\prime$ and $t$, one introduces the factor $\e^{-\epsilon/t}$ 
so that $\int_0^\infty\e^{-\epsilon/t}(\cdot)\dd t
=\int_\epsilon^\infty(\cdot)\dd t$ as $\epsilon\searrow0$. 
The latter holds true because $\e^{-\epsilon/t}t^n\lto0$ as $t\lto0$ 
for all $n\in\Z$, which further amounts to 
$\int_0^\epsilon\e^{-\epsilon/t}(\cdot)\dd t\lto0$ as 
$\epsilon\searrow0$
(see also the proof of proposition~\ref{prop:RcAlpha0}).

We compute \eqref{eq:HresProj-c-2} for $\alpha\geq0$ small.
%%%
\begin{prop}\label{prop:RcAlpha0}
For $\Im z\neq0$ and $\alpha=0$
\begin{align}
R_{\sigma^\prime\sigma}^c(z)(Q)=&
-\frac{\delta_{\sigma^\prime\sigma}}{8\pi^3Q^2}
\left(\delta_{s0}zK_2(\abs{Q}\sqrt{-z}) 
+\delta_{S1}[\delta_{s1}(z-2\beta)
K_2(\abs{Q}\sqrt{2\beta-z}) \right. 
\nonumber \\
&\left.+\delta_{s,-1}(z+2\beta)
K_2(\abs{Q}\sqrt{-2\beta-z})]\right)
\label{eq:Rca0}
\end{align}
for a.e. $Q\in\R^6$. If in addition $\beta=0$ then
\begin{equation}
R_{\sigma^\prime\sigma}^c(z)(Q)=
-\delta_{\sigma^\prime\sigma}
\frac{zK_2(\abs{Q}\sqrt{-z})}{8\pi^3Q^2}\,.
\label{eq:HresProj-c-alphabeta0}
\end{equation}
\end{prop}
%%%
($K_\nu$ is the Macdonald function of order $\nu$.) 
More generally, we have proposition~\ref{prop:RcDiagAlphaSmall}.
%%%
\begin{prop}\label{prop:RcDiagAlphaSmall}
For $\Im z\neq0$ and $\alpha\geq0$ arbitrarily small, 
the diagonal element
\begin{equation}
R_{\sigma\sigma}^c(z)=
R_{\sigma\sigma}^c(z)_{\alpha=0}+\alpha^2\Delta_\sigma(z)
+O(\alpha^4)
\label{eq:RczalphaSmall}
\end{equation}
where $R_{\sigma\sigma}^c(z)_{\alpha=0}$ is given by \eqref{eq:Rca0}, 
and the correction term
\begin{align}
\Delta_\sigma(z)(Q)=&\frac{\mp\img}{8\pi^3\beta^2\abs{Q}^3}
\Biggl(
\delta_{s0}[z^{3/2}K_3(\abs{Q}\sqrt{-z}) 
-\frac{1}{2}(z+2\beta)^{3/2}K_3(\abs{Q}\sqrt{-2\beta-z}) 
\nonumber \\
&-\frac{1}{2}(z-2\beta)^{3/2}K_3(\abs{Q}\sqrt{2\beta-z})
\mp\img(-1)^S\frac{x^-X^++x^+X^-}{8\abs{Q}} 
\nonumber \\
&\cdot(2z^2K_4(\abs{Q}\sqrt{-z})-(z+2\beta)^2
K_4(\abs{Q}\sqrt{-2\beta-z}) 
\nonumber \\
&-(z-2\beta)^2K_4(\abs{Q}\sqrt{2\beta-z}))] 
\nonumber \\
&+\delta_{S1}\{\delta_{s1}[(z-2\beta)^{3/2}
K_3(\abs{Q}\sqrt{2\beta-z})
-z^{3/2}K_3(\abs{Q}\sqrt{-z}) 
\nonumber \\
&\pm\img\beta\abs{Q}(z-2\beta)K_2(\abs{Q}\sqrt{2\beta-z})] 
\nonumber \\
&+\delta_{s,-1}[(z+2\beta)^{3/2}K_3(\abs{Q}\sqrt{-2\beta-z}) 
-z^{3/2}K_3(\abs{Q}\sqrt{-z}) 
\nonumber \\
&\mp\img\beta\abs{Q}(z+2\beta)K_2(\abs{Q}\sqrt{-2\beta-z})]\}
\Biggr)
\label{eq:DeltaCorr}
\end{align}
for a.e. $Q=(x,X)\in\R^6$; the upper (lower) sign is taken 
when $\Im z>0$ ($\Im z<0$). When in addition $\beta=0$, 
one assumes the limit $\beta\searrow0$ in \eqref{eq:DeltaCorr}.
\end{prop}
%%%
($X^\pm$ is defined similar to $x^\pm$.) 
For the purposes of the present paper, 
the computation of the diagonal elements of the two-particle 
Green function is sufficient. However, for completeness, 
we list non-diagonal elements in Appendix~\ref{sec:appC}.

We close the section by giving the proofs of the above propositions. 
%%%
\begin{proof}[Proof of proposition~\ref{prop:RcAlpha0}]
Using \eqref{eq:kerh-2} and \eqref{eq:atbt-small} 
\[G_{\img t}^{s^\prime s}(x)=\delta_{s^\prime s}
K_{\img t}^0(x)
(\cos(\beta t)-2\img s\sin(\beta t)) \quad 
\text{a.e. $x\in\R^3$}\]
and hence by \eqref{eq:HresProj-c-2}
\begin{align}
R_{\sigma^\prime\sigma}^c(z)(Q)=&
\pm\img\delta_{\sigma^\prime\sigma}
\int_0^\infty\e^{\pm\img t z}
K_{\pm\img t}^0(x)K_{\pm\img t}^0(X) 
\nonumber \\
&\cdot(\delta_{s0}+\delta_{S1}
[\delta_{s1}\e^{\mp2\img\beta t}
+\delta_{s,-1}\e^{\pm2\img\beta t}])\dd t
\label{eq:HresProj-c-2-alpha0}
\end{align}
for a.e. $Q=(x,X)\in\R^6$. The improper Riemann integrals
\begin{subequations}\label{eq:tm3-abc}
\begin{align}
\int_0^\infty\e^{\pm\img tz}
K_{\pm\img t}^0(x)K_{\pm\img t}^0(X)t^{-n}\dd t=&
\lim_{\epsilon\searrow0}\int_\epsilon^\infty \e^{\pm\img tz}
\frac{\exp(\frac{\pm\img Q^2}{4t})}{
(\pm4\pi\img t)^3}t^{-n} \dd t 
\label{eq:tm3-a} 
\\
=&\lim_{\epsilon\searrow0}
\int_0^\infty\e^{\pm\img tz}
\frac{\exp(\frac{\pm\img Q^2-\epsilon)}{4t})}{
(\pm4\pi\img t)^3}t^{-n} \dd t 
\label{eq:tm3-b} \\
=&\frac{\pm\img z^{1+n/2}}{
2^{3-n}\pi^3\abs{Q}^{n+2}}K_{n+2}(\abs{Q}\sqrt{-z})
\label{eq:tm3}
\end{align}
\end{subequations}
for $n\in\Z$; the upper (lower) sign is for $\Im z>0$ ($\Im z<0$). 
To show the second equality, let us define
\[I_\epsilon=\int_0^\infty
\e^{\img[tz+Q^2/(4t)]-\epsilon/t}t^n\dd t\]
for $\Im z>0$, $\abs{Q}>0$, $\epsilon>0$, $n\in\Z$. 
Making the substitution $t\lto t-\epsilon$ and then using 
the decomposition $\int_\epsilon^\infty
=\int_{\epsilon}^\delta+\int_\delta^\infty$ for $\delta>\epsilon$, 
we get 
$I_\epsilon=I_{\epsilon,\delta}^{(1)}+I_{\epsilon,\delta}^{(2)}$ 
where
\begin{align*}
I_{\epsilon,\delta}^{(1)}=&
\int_\epsilon^\delta
\e^{\img\{(t-\epsilon)z+Q^2/[4(t-\epsilon)]\}
-\epsilon/(t-\epsilon)}
(t-\epsilon)^n\dd t\,, 
\\
I_{\epsilon,\delta}^{(2)}=&
\int_\delta^\infty
\e^{\img\{(t-\epsilon)z+Q^2/[4(t-\epsilon)]\}
-\epsilon/(t-\epsilon)}
(t-\epsilon)^n\dd t\,.
\end{align*}
In the limit, the second integral
\begin{align*}
\lim_{\delta\searrow\epsilon}\lim_{\epsilon\searrow0}
I_{\epsilon,\delta}^{(2)}=&
\lim_{\delta\searrow\epsilon}\lim_{\epsilon\searrow0}
\int_\delta^\infty
\e^{\img\{(t-\epsilon)z+Q^2/[4(t-\epsilon)]\}
-\epsilon/(t-\epsilon)}
(t-\epsilon)^n\dd t 
\\
=&\lim_{\delta\searrow0}
\int_\delta^\infty
\e^{\img[tz+Q^2/(4t)]}t^n\dd t
\end{align*}
because $t\geq\delta>\epsilon$. The first integral
\[I_{\epsilon,\delta}^{(1)}=
\int_0^{\delta-\epsilon}
\e^{\img[tz+Q^2/(4t)]-\epsilon/t}t^n\dd t\,.\]
Since $\e^{-\epsilon/t}t^n\lto0$ as $t\lto0$, 
we have in the limit 
\[\lim_{\delta\searrow\epsilon}\lim_{\epsilon\searrow0}
I_{\epsilon,\delta}^{(1)}=0\]
which, together with the above result, shows the equality  
\eqref{eq:tm3-a} $=$ \eqref{eq:tm3-b} for $\Im z>0$; 
the proof for $\Im z<0$ is analogous. 

Relation~\eqref{eq:tm3} follows from \eqref{eq:tm3-b}
by using \eg \cite[equation~(9.42)]{Temme}.
Substitute \eqref{eq:tm3} with $n=0$ in 
\eqref{eq:HresProj-c-2-alpha0} and deduce \eqref{eq:Rca0}. 
Passing to the limit $\beta\searrow0$ one gets 
\eqref{eq:HresProj-c-alphabeta0}.
\end{proof}
%%%
\begin{proof}[Proof of proposition~\ref{prop:RcDiagAlphaSmall}]
We have by \eqref{eq:HresProj-c-2}
\begin{align}
R_{\sigma\sigma}^c(z)(x,X)=&
\pm\img\int_0^\infty\e^{\pm\img t z}
K_{\pm\img t}^0(x)K_{\pm\img t}^0(X)
(\delta_{s0}[a_{\pm\img t}^2-\beta^2b_{\pm\img t}^2 
\nonumber \\
&-(-1)^S\frac{\alpha^2}{8t^2}b_{\pm\img t}^2
(x^-X^++x^+X^-)] 
\nonumber \\
&+\delta_{S1}[\delta_{s1}
(a_{\pm\img t}-\beta b_{\pm\img t})^2
+\delta_{s,-1}(a_{\pm\img t}
+\beta b_{\pm\img t})^2])\dd t
\label{eq:1233444}
\end{align}
and by \eqref{eq:atbt-small}
\begin{subequations}\label{eq:atbt-alphasmall}
\begin{align}
a_{\img t}^2-\beta^2 b_{\img t}^2=&
1+\frac{\img\alpha^2}{2\beta^2t}(1-\cos(2\beta t))
+O(\alpha^4)\,, 
\\
(a_{\img t}\pm\beta b_{\img t})^2=&
\e^{\pm2\img\beta t}
+\frac{\alpha^2}{2\beta^2t}
(\e^{\pm2\img\beta t}(\img\pm 2\beta t)-\img)
+O(\alpha^4)
\end{align}
\end{subequations}
as $\alpha\searrow0$.
In view of \eqref{eq:HresProj-c-2-alpha0} and using 
$a_{-\img t}=\ol{a_{\img t}}$ (and similarly for $b_{-\img t}$) 
and the exponentiation of sine and cosine functions, substitute 
\eqref{eq:atbt-alphasmall} in \eqref{eq:1233444}, and
then apply \eqref{eq:tm3} for $n\in\{0,1,2\}$, and deduce 
\eqref{eq:RczalphaSmall} and \eqref{eq:DeltaCorr}.
\end{proof}
%%%%%%%%%%%%%%%%%%%%%%%%%%%%%%%%%%%%%%%%%%%%%%%%%%%%%%%%%%%%%%
%%%%%%%%%%%%%%%%%%%%%%%%%%%%%%%%%%%%%%%%%%%%%%%%%%%%%%%%%%%%%%
%%%%%%%%%%%%%%%%%%%%%%%%%%%%%%%%%%%%%%%%%%%%%%%%%%%%%%%%%%%%%%
%%%%%%%%%%%%%%%%%%%%%%%%%%%%%%%%%%%%%%%%%%%%%%%%%%%%%%%%%%%%%%
\section{Supersingular perturbation}
%%%%%%%%%%%%%%%%%%%%%%%%%%%%%%%%%%%%%%%%%%%%%%%%%%%%%%%%%%%%%%
%%%%%%%%%%%%%%%%%%%%%%%%%%%%%%%%%%%%%%%%%%%%%%%%%%%%%%%%%%%%%%
%%%%%%%%%%%%%%%%%%%%%%%%%%%%%%%%%%%%%%%%%%%%%%%%%%%%%%%%%%%%%%
%%%%%%%%%%%%%%%%%%%%%%%%%%%%%%%%%%%%%%%%%%%%%%%%%%%%%%%%%%%%%%
When considering the two atoms interacting via the zero-range 
potential which depends on the relative coordinate $x\in\R^3$, 
one follows a usual procedure and restricts the initial self-adjoint 
operator (Hamiltonian) $H^c$ to the set of functions vanishing at 
$x=0$, and then looks for possible self-adjoint extensions of the 
obtained symmetric operator. One should keep in mind that $H^c$ is 
non-separable in $Q$ for $\alpha>0$. 

Equivalently, one defines the singular distribution 
$\psi_\sigma$ ($\sigma\in\bas$) concentrated at $Q_0=(0,X)$ 
via the duality pairing
\begin{equation}
\braket{\psi_\sigma,f}=N_\sigma f_\sigma(Q_0)
\label{eq:del}
\end{equation}
for 
\[f=\sum_\sigma f_\sigma\ot\ket{\sigma}\in 
C^\infty(\R^6)\odot(\C^3\op\C^1)\]
and some normalization constant $N_\sigma>0$. 
Since $C_0^\infty\subset C^\infty$ and $C_0^\infty$ is 
dense in $H^2$, the restricted operator is thus the operator
$H^c$ subject to the boundary condition $\braket{\psi_\sigma,f}=0$ 
for $f\in\dom H^c$ and all $Q_0$. Using the scale 
$(\H_n=\H_n(H^c))_{n\in\Z}$ associated with $H^c$, 
$\dom H^c=\H_2$, \ie $H^c$ defines a mapping $\H_2\lto\H_0=\K^c$; 
the reader may refer to \cite{Albeverio00} for more details.

When $\psi_\sigma\in\H_{-n}\setm\H_0$ for some $n\in\N$,
the duality pairing in \eqref{eq:del} is equivalently 
defined via the scalar product $\braket{\cdot,\cdot}_0$ 
in $\H_0$:
\begin{equation}
\braket{\psi_\sigma,f}
=\braket{(\abs{\He^c}+\id)^{-n/2}\psi_\sigma\,,
(\abs{H^c}+\id)^{n/2}f}_0
\end{equation}
with $f\in\H_n$. Here $\id=\id_{\K^c}$ and 
$(\abs{\He^c}+\id)^{-n/2}$ is an extension of 
$(\abs{H^c}+\id)^{-n/2}$ when considered as a mapping from
$\H_{-n}$ onto $\H_0$. The duality pairing is well-defined 
since we have
$\abs{\braket{\psi_\sigma,f}}\leq 
\norm{\psi_\sigma}_{-n}\,\norm{f}_n$.
We remark the following:
%%%
\begin{prop}\label{prop:Hc-2}
The operator $\He^c$ is a continuation of $H^c\co\H_2\lto\H_0$ 
as a bounded operator from $\H_0$ into $\H_{-2}$. 
\end{prop}
%%%
\begin{proof}
Relation $\psi_\sigma\in\H_{-n}$ implies 
$h_\sigma=(\abs{\He^c}+\id)^{-n/2}\psi_\sigma\in\H_0$ 
and hence
\[\norm{\psi_\sigma}_{-n}=\norm{h_\sigma}_0
=\norm{(\abs{\He^c}+\id)^{-1}
(\abs{\He^c}+\id)h_\sigma}_0
=\norm{(\abs{\He^c}+\id)h_\sigma}_{-2}\]
\ie $\He^c$ defines a mapping $\H_0\lto\H_{-2}$.
\end{proof}
%%%
Thus, the task is to find $n\in\N$ for which 
$\psi_\sigma\in\H_{-n}$ holds. It suffices to verify the relation 
for $H^c$ parametrized by $\alpha=\beta=0$, because both $h^0$ and 
$h$ are self-adjoint on their common domain of definition; 
the Green function is given by \eqref{eq:HresProj-c-alphabeta0} for 
such $H^c$. On the other hand, in order to use the Krein formula 
\cite{Kurasov} for calculating eigenvalues later on 
(see also the discussion in section~\ref{sec:concl}), we have to 
compute the normalization constant $N_\sigma$, and we do so for 
$\alpha$ small (and $\beta$ arbitrary).
%%%
\begin{thm}\label{thm:super}
We have $\psi_\sigma\in\H_{-4}\setm\H_{-3}$. Moreover, 
if $((\He^c)^2+\id)^{-1}\psi_\sigma\in\H_0$ is the unit vector, 
then, for $\alpha\geq0$ arbitrarily small, 
the normalization constant satisfies the relation
\begin{align}
N_\sigma^{-2}=&
\frac{1}{512\pi^3}(\delta_{s0}\pi
+2\delta_{S1}[\delta_{s1}
(-2\beta+\theta_\beta(1+4\beta^2)) 
\nonumber \\ 
&+\delta_{s,-1}(2\beta+(\pi-\theta_\beta)
(1+4\beta^2))])
\nonumber \\ 
&+\frac{\alpha^2}{1536\pi^3\beta^2}(
\frac{1}{4}\delta_{s0}(2\beta
[\pi(5+4\beta^2)-11\theta_\beta
+4\beta(2-3\beta\theta_\beta)] 
\nonumber \\ 
&-4\log(1+4\beta^2)+\beta\sqrt{1+4\beta^2}
[(4\beta(\pi-\theta_\beta)-\log(1+4\beta^2))
\cos\theta_\beta 
\nonumber \\ 
&+2(\pi-\theta_\beta+\beta\log(1+4\beta^2))
\sin\theta_\beta ]) 
\nonumber \\ 
&+\delta_{S1}[\delta_{s1}
(8\beta^2(1-2\beta\theta_\beta)
+\log(1+4\beta^2)) 
\nonumber \\ 
&+\delta_{s,-1}(8\beta^2
(1+2\beta(\pi-\theta_\beta))
+\log(1+4\beta^2))] )
+O(\alpha^4)
\label{eq:fnrm}
\end{align}
with $\theta_\beta=\arg(2\beta+\img)$ 
($\arg$ is the principal value of the argument). When $\beta=0$, 
one assumes the limit $\beta\searrow0$ in \eqref{eq:fnrm}.
\end{thm}
%%%
\begin{rem}\label{rem:rem2}
For $\alpha=\beta=0$, relation \eqref{eq:fnrm} gives 
$N_\sigma=N$ and $\sqrt{N}=2c$, where $c=2\sqrt[4]{2}\sqrt{\pi}$ 
is the normalization constant for the functionals of class 
$\H_{-2}(h^0)$ \cite[section~2.3]{Albeverio00}.
\end{rem}
%%%
\begin{proof}
We use lemma~\ref{lem:infkdsl} to show that 
$\psi_\sigma\notin\H_{-3}$.
%%%
\begin{lem}\label{lem:infkdsl}
If $\phi\in\H_{-3}\setm\H_{-2}$ then 
\begin{align}
2^6\norm{\phi}_{-3}^2\geq&
\vert\langle\phi,\{\frac{3\img}{2}
[(\He^c+\img\id)^{-2}-(\He^c-\img\id)^{-2}] 
\nonumber \\
&+(\He^c+\img\id)^{-3}+(\He^c-\img\id)^{-3}\}
\phi\rangle\vert\,.
\label{eq:Hm3in}
\end{align}
\end{lem}
%%%
\begin{proof}
We have
\begin{align*}
&[(\He^c-\img\id)^{-1}+(\He^c+\img\id)^{-1}]\phi
=2\He^c((\He^c)^2+\id)^{-1}\phi\in\H_{-2}\,, \\
&[(\He^c-\img\id)^{-1}-(\He^c+\img\id)^{-1}]\phi
=2\img((\He^c)^2+\id)^{-1}\phi\in\H_0
\end{align*}
because $((\He^c)^2+\id)^{-1}\co\H_{-4}\lto\H_0$, 
$\H_{-3}\subset\H_{-4}$ densely, and 
$\He^c\co\H_0\lto\H_{-2}$ by proposition~\ref{prop:Hc-2}. 
Now, all we need is to apply
$(\abs{\He^c}+\id)^{-1}\geq
2^{-1}\abs{\He^c}((\He^c)^2+\id)^{-1}$ 
(see \eg the proof of theorem~3.1 in \cite{Albeverio97-1}) 
thrice to get
\begin{align*}
\norm{\phi}_{-3}^2\geq&
2^{-6}\vert\langle[(\He^c-\img\id)^{-1}
+(\He^c+\img\id)^{-1}]\phi,
(\He^c-\img\id)^{-2}\phi+(\He^c+\img\id)^{-2}\phi 
\\
&-\img[(\He^c-\img\id)^{-1}-(\He^c+\img\id)^{-1}]
\phi\rangle_0\vert\,.
\end{align*}
Transferring $[(\He^c-\img\id)^{-1}+(\He^c+\img\id)^{-1}]$ 
in the latter scalar product from left to right gives the 
result as claimed.
\end{proof}
%%%

Assume that $\psi_\sigma\in\H_{-3}\setm\H_{-2}$. 
In view of \eqref{eq:Hm3in}
\[\{\frac{3\img}{2}[(\He^c+\img\id)^{-2}
-(\He^c-\img\id)^{-2}]
+(\He^c+\img\id)^{-3}
+(\He^c-\img\id)^{-3}\}\psi_\sigma\in\H_3\,.\]
Since $\H_3\subset\H_0\subset\H_{-3}$ densely, for 
$f=\sum_\sigma f_\sigma\ot\ket{\sigma}\in\H_0$
\begin{align*}
&\langle\{\frac{3\img}{2}
[(\He^c+\img\id)^{-2}-(\He^c-\img\id)^{-2}]
+(\He^c+\img\id)^{-3}+(\He^c-\img\id)^{-3}\}
\psi_\sigma,f\rangle_0 
\\
&=\braket{\psi_\sigma,\{-\frac{3\img}{2}
[(\He^c-\img\id)^{-2}-(\He^c+\img\id)^{-2}]
+(\He^c-\img\id)^{-3}+(\He^c+\img\id)^{-3}\}f} 
\\
&=\braket{\psi_\sigma,\{-\frac{3\img}{2}
[R^c(\img)^\prime-R^c(-\img)^\prime]
+\frac{1}{2}R^c(\img)^{\prime\prime}
+\frac{1}{2}R^c(-\img)^{\prime\prime}\}f}
\end{align*}
where in the last step we also use the relations
\[R^c(z)^\prime
=\frac{\pd}{\pd w}R^c(w)\vert_{w=z}\,, 
\quad 
R^c(z)^{\prime\prime}
=\frac{\pd^2}{\pd w^2}R^c(w)\vert_{w=z}\]
for $z\in\res H^c$. Thus, 
by \eqref{eq:Hresz-c-2} and \eqref{eq:del}
\begin{align}
&\langle\{\frac{3\img}{2}
[(\He^c+\img\id)^{-2}-(\He^c-\img\id)^{-2}]
+(\He^c+\img\id)^{-3}+(\He^c-\img\id)^{-3}\}
\psi_\sigma,f\rangle_0 
\nonumber \\
&=N_\sigma\sum_{\sigma^\prime}\lim\int\{-\frac{3\img}{2}
[R_{\sigma\sigma^\prime}^c(\img)(Q_0-Q)^\prime
-R_{\sigma\sigma^\prime}^c(-\img)(Q_0-Q)^\prime] 
\nonumber \\
&\qquad\qquad\qquad
+\frac{1}{2}R_{\sigma\sigma^\prime}^c(\img)
(Q_0-Q)^{\prime\prime}
+\frac{1}{2}R_{\sigma\sigma^\prime}^c(-\img)
(Q_0-Q)^{\prime\prime}\}
f_{\sigma^\prime}(Q)\dd Q
\label{eq:fttttt0}
\end{align}
where 
\[R_{\sigma\sigma^\prime}^c(z)(Q)^\prime
=\frac{\pd}{\pd w}
R_{\sigma\sigma^\prime}^c(w)(Q)\vert_{w=z}\,, 
\quad 
R_{\sigma\sigma^\prime}^c(z)(Q)^{\prime\prime}
=\frac{\pd^2}{\pd w^2}
R_{\sigma\sigma^\prime}^c(w)(Q)\vert_{w=z}\]
for $\Im z\neq0$. But
\[\ol{R_{\sigma^\prime\sigma}^c(z)(Q)}
=R_{\sigma\sigma^\prime}^c(\ol{z})(-Q)\]
by \eqref{eq:HresProj-c-2}, and hence 
\begin{align}
&\{\frac{3\img}{2}[(\He^c+\img\id)^{-2}
-(\He^c-\img\id)^{-2}]
+(\He^c+\img\id)^{-3}
+(\He^c-\img\id)^{-3}\}\psi_\sigma 
\nonumber \\
&=\frac{N_\sigma}{2}\sum_{\sigma^\prime}
\{3\img[R_{\sigma^\prime\sigma}^c(-\img)
(\cdot-Q_0)^\prime
-R_{\sigma^\prime\sigma}^c(\img)
(\cdot-Q_0)^\prime] 
\nonumber \\
&\qquad\qquad
+R_{\sigma^\prime\sigma}^c(-\img)
(\cdot-Q_0)^{\prime\prime}
+R_{\sigma^\prime\sigma}^c(\img)
(\cdot-Q_0)^{\prime\prime}\}
\ot\ket{\sigma^\prime}
\label{eq:fttttt}
\end{align}
a.e. on $\R^6$. When deriving \eqref{eq:fttttt} 
from \eqref{eq:fttttt0} we have also used the following property: 
Since the Lebesgue integral $\int_{\R^6}$ on the left-hand side of 
\eqref{eq:fttttt0} exists by hypothesis on $\psi_\sigma$, 
it coincides with the improper Riemann integral $\lim\int\equiv
\lim_{R\nearrow\infty}\int_{\abs{Q}\leq R}$.

By \eqref{eq:del} and \eqref{eq:fttttt}, 
the duality pairing on the right-hand side of \eqref{eq:Hm3in} 
with $\phi=\psi_\sigma$ is therefore given by 
\begin{equation}
\frac{N_\sigma^2}{2}\lim_{\abs{Q}\searrow0}
(3\img[R_{\sigma\sigma}^c(-\img)(Q)^\prime
-R_{\sigma\sigma}^c(\img)(Q)^\prime]
+R_{\sigma\sigma}^c(-\img)(Q)^{\prime\prime}
+R_{\sigma\sigma}^c(\img)(Q)^{\prime\prime})\,.
\label{eq:fttttt1}
\end{equation}
It suffices to take $R^c_{\sigma\sigma}(\pm\img)$ as in 
\eqref{eq:HresProj-c-alphabeta0}
to show the non-existence of \eqref{eq:fttttt1}:
\begin{align*}
3\img&[R_{\sigma\sigma}^c(-\img)(Q)^\prime
-R_{\sigma\sigma}^c(\img)(Q)^\prime]
+R_{\sigma\sigma}^c(-\img)(Q)^{\prime\prime}
+R_{\sigma\sigma}^c(\img)(Q)^{\prime\prime} 
\\
&=O(\log\abs{Q})\quad\text{as}\quad\abs{Q}\searrow0
\end{align*}
\ie $\psi_\sigma\notin\H_{-3}$.

Assume that $\psi_\sigma\in\H_{-4}$. Then
\[\norm{\psi_\sigma}_{-4}\leq
\norm{\psi_\sigma}_{-4}^*
=\norm{((\He^c)^2+\id)^{-1}\psi_\sigma}_0
=\frac{1}{2}\norm{[(\He^c-\img\id)^{-1}
-(\He^c+\img\id)^{-1}]
\psi_\sigma}_0\,.\]
By rearranging the terms within the equivalent norm 
$\norm{\psi_\sigma}_{-4}^*$ we get that 
\[\norm{\psi_\sigma}_{-4}^{*\,2}
=\frac{1}{4}\braket{\psi_\sigma,
\{\img[(\He^c+\img\id)^{-1}-(\He^c-\img\id)^{-1}]
-(\He^c+\img\id)^{-2}
-(\He^c-\img\id)^{-2}\}\psi_\sigma}\,.\]
Then, repeating the steps that were used for obtaining 
\eqref{eq:fttttt} we get that
\[\norm{\psi_\sigma}_{-4}^{*\,2}
=\frac{N_\sigma^2}{4}\lim_{\abs{Q}\searrow0}
(\img[R_{\sigma\sigma}^c(-\img)(Q)
-R_{\sigma\sigma}^c(\img)(Q)]
-R_{\sigma\sigma}^c(-\img)(Q)^\prime
-R_{\sigma\sigma}^c(\img)(Q)^\prime)\,.\]
Again, taking $R^c_{\sigma\sigma}(\pm\img)$ as in 
\eqref{eq:HresProj-c-alphabeta0}, we get that 
$\norm{\psi_\sigma}_{-4}^*=N_\sigma/(16\sqrt{2}\pi)$; hence 
$\psi_\sigma\in\H_{-4}$. 
Using proposition~\ref{prop:RcDiagAlphaSmall}, 
the above formula gives \eqref{eq:fnrm}.
This accomplishes the proof of the theorem.
\end{proof}
%%%%%%%%%%%%%%%%%%%%%%%%%%%%%%%%%%%%%%%%%%%%%%%%%%%%%%%%%%%%%%
%%%%%%%%%%%%%%%%%%%%%%%%%%%%%%%%%%%%%%%%%%%%%%%%%%%%%%%%%%%%%%
%%%%%%%%%%%%%%%%%%%%%%%%%%%%%%%%%%%%%%%%%%%%%%%%%%%%%%%%%%%%%%
%%%%%%%%%%%%%%%%%%%%%%%%%%%%%%%%%%%%%%%%%%%%%%%%%%%%%%%%%%%%%%
\section{Concluding remarks and discussion}\label{sec:concl}
%%%%%%%%%%%%%%%%%%%%%%%%%%%%%%%%%%%%%%%%%%%%%%%%%%%%%%%%%%%%%%
%%%%%%%%%%%%%%%%%%%%%%%%%%%%%%%%%%%%%%%%%%%%%%%%%%%%%%%%%%%%%%
%%%%%%%%%%%%%%%%%%%%%%%%%%%%%%%%%%%%%%%%%%%%%%%%%%%%%%%%%%%%%%
%%%%%%%%%%%%%%%%%%%%%%%%%%%%%%%%%%%%%%%%%%%%%%%%%%%%%%%%%%%%%%
In the paper, in theorem~\ref{thm:lE} and corollary~\ref{cor:lE}, 
we present the integral kernel of the one-parameter 
unitary group for the Rashba spin-orbit coupled operator in 
dimension three. The main motive for considering the unitary group 
is to derive the Green function, \eqref{eq:Hresz-c-2} and 
\eqref{eq:HresProj-c-2}, for the corresponding two-particle operator, 
which is necessary for the spectral analysis of spin-orbit coupled cold 
molecules. For $\alpha\geq0$ small, we compute explicitly the elements 
of the two-particle Green function in propositions~\ref{prop:RcAlpha0}, 
\ref{prop:RcDiagAlphaSmall}, \ref{prop:RcNonDiagAlphaSmall}.

The inter-atomic interaction is zero-range and therefore 
we apply the singular perturbation theory. We show that, 
since the two-particle Hamiltonian in the center-of-mass coordinate 
system is non-separable for $\alpha>0$, the perturbation associated 
to the total Hamiltonian is supersingular (theorem~\ref{thm:super}). 
As a result, no self-adjoint operator can be constructed for 
describing the formation of spin-orbit coupled molecules with 
point-interaction. Instead, one considers the so-called regular 
operators whose spectrum is known to be pure real. 

For example, assume that $\beta=0$ and $\alpha$ is so small that 
we can practically put $\alpha=0$. Formally, the problem reduces 
to the analysis of the operator $-2\Delta_x-\frac{1}{2}\Delta_X$ 
(which is separable) plus the $x$-dependent singular perturbation. 
Using \eqref{eq:HresProj-c-alphabeta0}, the resolvent formula in 
\cite{Kurasov}, and the normalization constant $16\sqrt{2}\pi$
(see \eqref{eq:fnrm}), the singular points $\lambda$ of the 
restricted (to the original Hilbert space) two-particle resolvent 
of the one-parameter regular operator satisfy the relation
\[0=(1+\lambda)[2(1+3\lambda)+\pi(\gamma(1+\lambda)-1)]
-4\lambda^2\log(-\lambda)\] 
for some non-uniquely defined real parameter $\gamma$; hence 
$\lambda<0$ necessarily. On the other hand, when we associate 
the perturbation to $-2\Delta_x$, we have the two-particle 
case described in \cite[theorems~5.2.1 and 5.2.2]{Albeverio00}.
Namely, one solves the eigenvalue problem for the single-particle 
operator $-2\Delta+2c\delta$, with $c$ as in remark~\ref{rem:rem2},
for which it is well-known that there is the single eigenvalue 
below $0$. The question, which was raised in
\cite{Kurasov} in a much more general setting, is whether there 
exists the similarity operator that transforms the  
non-self-adjoint case to the self-adjoint one.
%%%%%%%%%%%%%%%%%%%%%%%%%%%%%%%%%%%%%%%%%%%%%%%%%%%%%%%%%%%%%%
%%%%%%%%%%%%%%%%%%%%%%%%%%%%%%%%%%%%%%%%%%%%%%%%%%%%%%%%%%%%%%
%%%%%%%%%%%%%%%%%%%%%%%%%%%%%%%%%%%%%%%%%%%%%%%%%%%%%%%%%%%%%%
%%%%%%%%%%%%%%%%%%%%%%%%%%%%%%%%%%%%%%%%%%%%%%%%%%%%%%%%%%%%%%
\appendix
\section{Non-diagonal elements of Green function}\label{sec:appC}
\setcounter{section}{1}
%%%%%%%%%%%%%%%%%%%%%%%%%%%%%%%%%%%%%%%%%%%%%%%%%%%%%%%%%%%%%%
%%%%%%%%%%%%%%%%%%%%%%%%%%%%%%%%%%%%%%%%%%%%%%%%%%%%%%%%%%%%%%
%%%%%%%%%%%%%%%%%%%%%%%%%%%%%%%%%%%%%%%%%%%%%%%%%%%%%%%%%%%%%%
%%%%%%%%%%%%%%%%%%%%%%%%%%%%%%%%%%%%%%%%%%%%%%%%%%%%%%%%%%%%%%
Here we list non-diagonal elements $R_{\sigma^\prime\sigma}^c(z)$ 
of the two-particle Green function up to $O(\alpha^4)$; recall 
\eqref{eq:HresProj-c-2}.
%%%
\begin{prop}\label{prop:RcNonDiagAlphaSmall}
For $\Im z\neq0$ and $\alpha\geq0$ arbitrarily small, 
the non-diagonal element ($\sigma^\prime\neq\sigma$)
\begin{equation}
R_{\sigma^\prime\sigma}^c(z)=
\alpha\Delta_{\sigma^\prime\sigma}^{(1)}(z)
+\alpha^2\Delta_{\sigma^\prime\sigma}^{(2)}(z)
+\alpha^3\Delta_{\sigma^\prime\sigma}^{(3)}(z)
+O(\alpha^4)
\label{eq:RczalphaSmallNonD}
\end{equation}
where 
\begin{subequations}\label{eq:Delta123}
\begin{align}
\Delta_{\sigma^\prime\sigma}^{(1)}(z)(Q)=&
\frac{\pm\img}{16\sqrt{2}\pi^3\beta\abs{Q}^3}
\Biggl(
(\delta_{s^\prime1}\delta_{S^\prime1}\delta_{s0}
[X^--(-1)^Sx^-] 
\nonumber \\
&-\delta_{s^\prime0}\delta_{s1}\delta_{S1}
[X^+-(-1)^{S^\prime}x^+])
[z^{3/2}K_3(\abs{Q}\sqrt{-z}) 
\nonumber \\
&-(z-2\beta)^{3/2}K_3(\abs{Q}\sqrt{2\beta-z})] 
\nonumber \\
&+(\delta_{s^\prime0}\delta_{s,-1}\delta_{S1}
[(-1)^{S^\prime}X^--x^-]
-\delta_{s^\prime,-1}\delta_{S^\prime1}\delta_{s0}
[(-1)^SX^+-x^+]) 
\nonumber \\
&\cdot[z^{3/2}K_3(\abs{Q}\sqrt{-z})
-(z+2\beta)^{3/2}K_3(\abs{Q}\sqrt{-2\beta-z})]
\Biggr)\,,
\\
\Delta_{\sigma^\prime\sigma}^{(2)}(z)(Q)=&
-\frac{1}{64\pi^3\beta^2\abs{Q}^4}
(2\delta_{S^\prime S}\delta_{S1}
[\delta_{s^\prime1}\delta_{s,-1}x^-X^-
+\delta_{s^\prime,-1}\delta_{s1}x^+X^+] 
\nonumber \\
&+(-1)^S\delta_{s^\prime s}\delta_{s0}
(1-\delta_{S^\prime S})[x^-X^+-x^+X^-])
[2z^2K_4(\abs{Q}\sqrt{-z}) 
\nonumber \\
&-(z+2\beta)^2K_4(\abs{Q}\sqrt{-2\beta-z})
-(z-2\beta)^2K_4(\abs{Q}\sqrt{2\beta-z})]\,, 
\\
\Delta_{\sigma^\prime\sigma}^{(3)}(z)(Q)=&
\frac{1}{32\sqrt{2}\pi^3\beta^3Q^4}
\Biggl(
(\delta_{s^\prime1}\delta_{S^\prime1}
\delta_{s0}[X^--(-1)^Sx^-] 
\nonumber \\
&-\delta_{s^\prime0}\delta_{s1}\delta_{S1}
[X^+-(-1)^{S^\prime}x^+])
[3(z-2\beta)^2K_4(\abs{Q}\sqrt{2\beta-z}) 
\nonumber \\
&-4z^2K_4(\abs{Q}\sqrt{-z})+(z+2\beta)^2
K_4(\abs{Q}\sqrt{-2\beta-z}) 
\nonumber \\
&\pm2\img\beta\abs{Q}(z-2\beta)^{3/2}
K_3(\abs{Q}\sqrt{2\beta-z})] 
\nonumber \\
&+(\delta_{s^\prime0}\delta_{s,-1}\delta_{S1}
[(-1)^{S^\prime}X^--x^-]
-\delta_{s^\prime,-1}\delta_{S^\prime1}\delta_{s0}
[(-1)^SX^+-x^+]) 
\nonumber \\
&\cdot
[3(z+2\beta)^2K_4(\abs{Q}\sqrt{-2\beta-z})
-4z^2K_4(\abs{Q}\sqrt{-z}) 
\nonumber \\
&+(z-2\beta)^2K_4(\abs{Q}\sqrt{2\beta-z}) 
\nonumber \\
&\mp2\img\beta\abs{Q}(z+2\beta)^{3/2}
K_3(\abs{Q}\sqrt{-2\beta-z})]
\Biggr)
\end{align}
\end{subequations}
for a.e. $Q=(x,X)\in\R^6$; the upper (lower) sign is taken when 
$\Im z>0$ ($\Im z<0$).
When in addition $\beta=0$, one assumes the limit $\beta\searrow0$ 
in the above expressions.
\end{prop}
%%%
\begin{proof}
We have by \eqref{eq:HresProj-c-2}
\begin{align}
R_{\sigma^\prime\sigma}^c(z)(x,X)=&
\int_0^\infty\e^{\pm\img t z}
K_{\pm\img t}^0(x)K_{\pm\img t}^0(X) 
\nonumber \\
&\cdot\Biggl(
\frac{\alpha}{2\sqrt{2}t}b_{\pm\img t}
\{(a_{\pm\img t}-\beta b_{\pm\img t})
(\delta_{s^\prime1}\delta_{S^\prime1}\delta_{s0}
[X^--(-1)^Sx^-] 
\nonumber \\
&-\delta_{s^\prime0}\delta_{s1}\delta_{S1}
[X^+-(-1)^{S^\prime}x^+])
-(a_{\pm\img t}+\beta b_{\pm\img t}) 
\nonumber \\
&\cdot(\delta_{s^\prime0}\delta_{s,-1}\delta_{S1}
[(-1)^{S^\prime}X^--x^-]
-\delta_{s^\prime,-1}\delta_{S^\prime1}\delta_{s0}
[(-1)^SX^+-x^+])\} 
\nonumber \\
&\mp\frac{\img\alpha^2}{8t^2}b_{\pm\img t}^2
\{2\delta_{S^\prime S}\delta_{S1}
[\delta_{s^\prime1}\delta_{s,-1}x^-X^-
+\delta_{s^\prime,-1}\delta_{s1}x^+X^+] 
\nonumber \\
&+(-1)^S\delta_{s^\prime s}\delta_{s0}
(1-\delta_{S^\prime S})[x^-X^+-x^+X^-]\}
\Biggr)\dd t
\label{eq:opapa}
\end{align}
for $\sigma^\prime\neq\sigma$. Using $a_{-\img t}=\ol{a_{\img t}}$ 
(and similarly for $b_{-\img t}$) and the exponentiation
of sine and cosine functions, substitute \eqref{eq:atbt-small} 
in \eqref{eq:opapa}, and then apply \eqref{eq:tm3} for $n\in\{1,2\}$, 
and deduce \eqref{eq:RczalphaSmallNonD} and \eqref{eq:Delta123}.
\end{proof}
%%%
%\bibliographystyle{unsrt}
%\bibliography{BiB}

\begin{thebibliography}{10}

\bibitem{Simon-2}
B.~Simon.
\newblock {\em Functional Integration and Quantum Physics}.
\newblock AMS Chelsea Publishing, Providence, Rhode Island, 2 edition, 2005.

\bibitem{Yosida}
K.~Yosida.
\newblock {\em Functional Analysis}.
\newblock Springer-Verlag, Berlin Heidelberg New York, 6 edition, 1995.

\bibitem{Simon-1}
B.~Simon.
\newblock Schr\"{o}dinger semigroups.
\newblock {\em Bulletin Amer. Math. Soc.}, 7(3):447--526, 1982.

\bibitem{Simon-3}
B.~Simon and R.~Hoegh-Krohn.
\newblock Hypercontractive semigroups and two dimensional self-coupled {B}ose
  fields.
\newblock {\em J. Func. Anal.}, 9:121--180, 1972.

\bibitem{Carlone11}
Raffaelle Carlone and Pavel Exner.
\newblock Dynamics of an electron confined to a "hybrid plane" and interacting
  with a magnetic field.
\newblock {\em Rep. Math. Phys.}, 67(2):211--227, 2011.

\bibitem{Cacciapuoti09}
C.~Cacciapuoti, R.~Carlone, and R.~Figari.
\newblock Resonances in models of spin-dependent point interactions.
\newblock {\em J. Phys. A: Math. Theor.}, 42(3):035202, 2009.

\bibitem{Cacciapuoti07}
C.~Cacciapuoti, R.~Carlone, and R.~Figari.
\newblock Spin-dependent point potentials in one and three dimensions.
\newblock {\em J. Phys. A: Math. Theor.}, 40(2):249--261, 2007.

\bibitem{Bruning07}
Jochen Br\"{u}ning, Vladimir Geyler, and Konstantin Pankrashkin.
\newblock Explicit {G}reen functions for spin-orbit {H}amiltonians.
\newblock {\em J. Phys. A: Math. Theor.}, 40:F697--F704, 2007.

\bibitem{Exner07}
P.~Exner and P.~\v{S}eba.
\newblock A "{H}ybrid {P}lane" with {S}pin-{O}rbit {I}nteraction.
\newblock {\em Russian J. Math. Phys.}, 14(4):430--434, 2007.

\bibitem{Jursenas16}
R.~Jur\v{s}\.{e}nas.
\newblock Spectrum of a family of spin-orbit coupled {H}amiltonians with
  singular perturbation.
\newblock {\em J. Phys. A: Math. Theor.}, 49(6):065202, 2016.

\bibitem{Jursenas14}
R.~Jur\v{s}\.{e}nas.
\newblock Series expansion for the {F}ourier transform of a rational function
  in three dimensions.
\newblock {\em Rep. Math. Phys.}, 75(1):1--24, 2015.

\bibitem{Fu}
Z.~Fu, L.~Huang, Z.~Meng, P.~Wang, L.~Zhang, S.~Zhang, H.~Zhai, P.~Zhang, and
  J.~Zhang.
\newblock Production of {F}eshbach molecules induced by spin-orbit coupling in
  {F}ermi gases.
\newblock {\em Nature Physics}, 10(2):110--115, 2013.

\bibitem{Albeverio00}
S.~Albeverio and P.~Kurasov.
\newblock {\em Singular Perturbations of Differential Operators}.
\newblock London Mathematical Society Lecture Note Series 271. Cambridge
  University Press, UK, 2000.

\bibitem{Kurasov2}
P.~Kurasov.
\newblock Triplet extensions {I}: {S}emibounded operators in the scale of
  {H}ilbert spaces.
\newblock {\em Journal d'Analyse Mathematique}, 107(1):252--286, 2009.

\bibitem{Kurasov}
P.~Kurasov.
\newblock $\mathcal{H}_{-n}$-perturbations of self-adjoint operators and
  {K}rein's resolvent formula.
\newblock {\em Integr. Equ. Oper. Theory}, 45(4):437--460, 2003.

\bibitem{Simon-4}
M.~Reed and B.~Simon.
\newblock Tensor products of closed operators on {B}anach spaces.
\newblock {\em J. Func. Anal.}, 13(2):107--124, 1973.

\bibitem{Srivastava}
H.~M. Srivastava and H.~L. Manocha.
\newblock {\em A Treatise on Generating Functions}.
\newblock Ellis Horwood Limited, New York, 1984.

\bibitem{Debiard}
A.~Debiard and B.~Gaveau.
\newblock Hypergeometric symbolic calculus. {II} -- {S}ystems of confluent
  equations.
\newblock {\em Bull. Sci. Math.}, 127(3):261--280, 2003.

\bibitem{Jucys}
A.~P. Jucys and A.~A. Bandzaitis.
\newblock {\em Theory of Angular Momentum in Quantum Mechanics}.
\newblock Mokslas Publishers, Vilnius (in Russian), 1977.

\bibitem{Temme}
N.~M. Temme.
\newblock {\em Special Functions. An Introduction to the Classical Functions of
  Mathematical Physics}.
\newblock John Wiley \& Sons, Inc., New York, 1996.

\bibitem{Albeverio97-1}
S.~Albeverio and P.~Kurasov.
\newblock Rank one perturbations of not semibounded operators.
\newblock {\em Integr. Equ. Oper. Theory}, 27:379--400, 1997.

\end{thebibliography}

\end{document}